\documentclass{article}
\usepackage[english]{babel}
\usepackage{tikz,tikz-3dplot}
\usetikzlibrary{matrix}
\usepackage{amssymb}
\usepackage[utf8]{inputenc}
\usepackage{calrsfs}
\usepackage{amsmath}
\usepackage{amsthm}
\usepackage{enumerate}
\usepackage{caption}
\usepackage{algorithm}
\usepackage{algpseudocode}
\usepackage{algorithmicx}
\usepackage{graphicx}
\usepackage{a4wide}
\usetikzlibrary[shapes.misc]
\usetikzlibrary[shapes.geometric]
\newtheorem{de}{Definition}[section]
\newtheorem{theo}{Theorem}[section]
\newtheorem{prop}[theo]{Proposition}
\newtheorem{cor}[theo]{Corollary}

\title{Leader Election And Local Identifiers For 3D Programmable Matter}
\author{Nicolas Gastineau \and Wahabou Abdou \and Nader Mbarek \and Olivier Togni\\
LIB, Université Bourgogne Franche-Comt\'e, Dijon, France}
\date{\today}

\begin{document}
\maketitle

\begin{abstract}

In this paper, we present two deterministic leader election algorithms for programmable matter on the face-centered cubic grid. The face-centered cubic grid is a 3-dimensional 12-regular infinite grid that represents an optimal way to pack spheres (i.e., spherical particles or modules in the context of the programmable matter) in the 3-dimensional space. While the first leader election algorithm requires a strong hypothesis about the initial configuration of the particles and no hypothesis on the system configurations that the particles are forming, the second one requires fewer hypothesis about the initial configuration of the particles but does not work for all possible particles' arrangement.
We also describe a way to compute and assign $\ell$-local identifiers to the particles in this grid with a memory space not dependent on the number of particles.  A $\ell$-local identifier is a variable assigned to each particle in such a way that particles at distance at most $\ell$ each have a different identifier.
\end{abstract}

\section{Introduction}

Programmable matter consists of modular robots (called modules or particles) able to fix to adjacent modules and send (receive) messages to (from) other modules fixed to the entity. Thus, the different modules form a geometric shape which is a network.
Usually, a module can fix to another module with a finite number of ports (see Figure~\ref{sphere} for an example of spherical modules). Also, the modules know the ports that are in contact with other modules and have a knowledge about the geographic position of their ports. Moreover, the ports are supposed to be homogeneously distributed along the surface of each module. Such assumptions imply that the way how the modules are connected can be modeled by a two or three dimensional grid (depending if you consider the modules on a two dimensional plane or in the 3-dimensional space).
The geometric amoebot model~\cite{JJD2017,ZD2014,ZD2015,ZD2015b,ZD2016a,ZD2016b,ZD2017,DIL2017,EM2019} aims to specify the properties of a network for programmable matter on a plane. In the geometric amoebot model, the considered grid is the regular triangular grid.
In the three dimensional context, another grid should be considered. One natural choice of grid is the face-centered cubic grid.

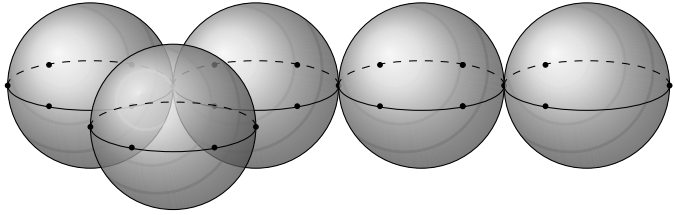
\begin{figure}[t]
\begin{center}
\begin{tikzpicture}[scale=0.55]
  \shade[ball color = gray!40, opacity = 0.4] (0,0) circle (2cm);
  \draw (0,0) circle (2cm);
  \draw (-2,0) arc (180:360:2 and 0.6);
  \draw[dashed] (2,0) arc (0:180:2 and 0.6);
  \fill[fill=black] (2,0) circle (2pt);
  \fill[fill=black] (-2,0) circle (2pt);
  \fill[fill=black] (-1,-0.5) circle (2pt);
  \fill[fill=black] (1,-0.5) circle (2pt);
  \fill[fill=black] (-1,0.5) circle (2pt);
  \fill[fill=black] (1,0.5) circle (2pt);
  \shade[ball color = gray!40, opacity = 0.4] (4,0) circle (2cm);
  \draw (4,0) circle (2cm);
  \draw (2,0) arc (180:360:2 and 0.6);
  \draw[dashed] (6,0) arc (0:180:2 and 0.6);
  \fill[fill=black] (6,0) circle (2pt);
  \fill[fill=black] (-1+4,-0.5) circle (2pt);
  \fill[fill=black] (1+4,-0.5) circle (2pt);
  \fill[fill=black] (-1+4,0.5) circle (2pt);
  \fill[fill=black] (1+4,0.5) circle (2pt);
  \shade[ball color = gray!40, opacity = 0.4] (8,0) circle (2cm);
  \draw (8,0) circle (2cm);
  \draw (6,0) arc (180:360:2 and 0.6);
  \draw[dashed] (10,0) arc (0:180:2 and 0.6);
  \fill[fill=black] (10,0) circle (2pt);
  \fill[fill=black] (-1+8,-0.5) circle (2pt);
  \fill[fill=black] (1+8,-0.5) circle (2pt);
  \fill[fill=black] (-1+8,0.5) circle (2pt);
  \fill[fill=black] (1+8,0.5) circle (2pt);
  \shade[ball color = gray!40, opacity = 0.4] (12,0) circle (2cm);
  \draw (12,0) circle (2cm);
  \draw (10,0) arc (180:360:2 and 0.6);
  \draw[dashed] (14,0) arc (0:180:2 and 0.6);
  \fill[fill=black] (14,0) circle (2pt);
  \fill[fill=black] (-1+10,-0.5) circle (2pt);
  \fill[fill=black] (1+10,-0.5) circle (2pt);
  \fill[fill=black] (-1+10,0.5) circle (2pt);
  \fill[fill=black] (1+10,0.5) circle (2pt);
  
  \shade[ball color = gray!40, opacity = 0.4] (2,-0.5*2) circle (2cm);
  \draw (2,-0.5*2) circle (2cm);
  \draw (0,-0.5*2) arc (180:360:2 and 0.6);
  \draw[dashed] (4,-0.5*2)  arc (0:180:2 and 0.6);
    \fill[fill=black] (1,-1.5) circle (2pt);
  \fill[fill=black] (3,-1.5) circle (2pt);
    \fill[fill=black] (0,-1) circle (2pt);
  \fill[fill=black] (4,-1) circle (2pt);
\end{tikzpicture}
\caption{Five spherical particles forming a simple structure (circle: port of the particle).}
\label{sphere}
\end{center}
\end{figure}

The face-centered cubic (FCC) grid is a grid containing multiple copies of triangular and square grids and can be an alternative to represent objects in the three dimensional space using $\mathbb{Z}^2$. This grid has been studied in different scientific areas including crystallography and visualization \cite{BJO1990}. It can be noted that this grid is also called cannonball grid since it represents one way to fill the three dimensional space with cannonballs (spheres) of the same size, while optimizing the density (the vertex set being the cannonballs and the edge set being the physical contacts between cannonballs). Analogously, in a 2-dimensional space, the regular triangular grid also represents a way to fill, in an optimal way, the space with spheres (or disks) of the same size (the graph is obtained from the 2-dimensional space using the same construction).

Since programmable matter is a scientific and technological challenging area, several projects aim to build programmable matter prototypes. One of such projects \cite{BOUR2018,TU2018}, financed by the French National Agency for Research, aims to build quasi-spherical particles able to deform them-selves in order to move.
These particles have a kind of cuboctahedron's form with twelve ports. The way how the ports are distributed among the surface of a particle implies a face-centered cubic structure for the network of particles. 
The final goal of this project is to sculpt a shape-memory polymer sheet with programmable matter.
In the continuity of the algorithm phase of this project \cite{GT2018,BOUR2018}, we propose algorithms for the self-configuration for these types of prototypes.

In the context of programmable matter inducing a face-centered cubic grid, we make the following hypothesis. We suppose that each particle is manufactured in the same way and consequently that every particle has, initially, its ports labeled in the same way. However, since the particles can be displaced, it is possible that some particles have been rotated since their conception. Thus, we consider two different models depending on the fact that particles have been rotated or not.
In the homogeneous case, we consider that the particles have their ports labeled exactly in the same way, i.e., ports number $i$ of any two particles are oriented in the same cardinal direction for every $i$.
In the heterogeneous case, we consider that the particles have their ports labeled differently but for every two particles $p$ and $p'$ there exists a spatial rotation on the particle $p$ such that the two particles have their ports labeled exactly in the same way.

Distributed algorithms aim to give a theoretical algorithmic framework in order to model the execution of an algorithm running on a network of computational elements that can cooperate in order to solve network problems.
In distributed algorithm frameworks, it is often supposed that the different elements of the network do not have a unique identity, i.e., the network is anonymous. In anonymous networks, a natural question is how to perform a leader election, i.e., how to determine a singular element in an anonymous network.
It is well known that for some network structures, the ring for example, there is no deterministic leader election algorithm~\cite{AW2004}. However, since in the field of programmable matter the ports of the particles are supposed to be labeled following some hypothesis, the class of graphs for which there exists a deterministic leader election in the context of programmable matter is larger than in the general context.

\paragraph{Related work.} 
Leader election is a classical and well studied problem in distributed systems (see~\cite[Section 3]{AW2004}). In the context of programmable matter, the contributions are less numerous, more recent and only concern (to our knowledge) 2-dimensional shapes.
Derakhshandeh et al.\ \cite{ZD2015b} proposed a randomized leader election algorithm in the geometric amoebot model. 
Recently, Daymude et al.\ \cite{JJD2017} have improved the algorithm from Derakhshandeh et al.\ \cite{ZD2015b} by giving stronger theoretical guarantees. Also, very recently, Di Luna et al.\ \cite{DIL2017} have introduced a leader election algorithm called consumption algorithm that consists, like in~\cite{GT2018}, in successively removing the candidacy of the particles on the border. This algorithm works for hole-free system configurations.
Moreover, Bazzi and Briones~\cite{BAZ2019} have recently established a stationary and deterministic leader election algorithm for the geometric amoebot model in which the system forms a simply connected shape working if and only if the leader election is possible under their assumptions. 
Concurrently with Bazzi and Briones's paper,  Emek et al.\ \cite{EM2019} have also proposed a deterministic leader election algorithm for the amoebot model, using the ability of the amoebots to move. Their paper also contains a useful comparison of existing leader election algorithms for the amoebot model. Also, D'angelo et al. \cite{DA2020} have recently established a deterministic leader election in the case the particles are not able to communicate but know the particles present in their neighborhood and their states (they suppose that the particles can only be in two different states).

\paragraph{Contributions and organization of the paper.}
In this paper we consider two main problems for programmable matter in the 3-dimensional space. The first one is leader election and the second one is local identifiers assignment. For leader election, we propose two deterministic algorithms (Algorithms~\ref{alg1} and \ref{alg2}). The first one requires an initial configuration of the particles such that they all have their ports numbered in the same directions and no hypothesis on the shape that the particles are forming. The second one requires no system configurations but works only for some shapes. Also, the first algorithm requires a $O(|S|\log(|S|))$ memory space ($|S|$ being the number of particles) whereas the second algorithm only needs a constant memory space.

The first algorithm consists, for each particle, in computing a local description of the geographical structure formed by the particles and, then, by electing the particle which has a maximum position, considering the lexicographical order.

The second algorithm consists in successively removing, from a set $S$ of potential leaders, some particles which are both on the geographical border of the subgraph of the FCC grid induced by $S$ and not an articulation of it (an articulation is a vertex whose removal split the network in two connected components).

As concerning the problem of assigning $\ell$-local identifiers, i.e., identifiers such that particles at distance at most $\ell$ have distinct identifiers, we propose an algorithm (Algorithm~\ref{alg4}) that assigns identifiers based on a leader election algorithm and a coloring of the $\ell$-th power of the FCC grid. This algorithm requires a memory space of $O(\log(\ell^3))$ for each particle (hence not dependent on the number of particles).

The rest of the paper is organized as follows. Notation and definitions related to graphs and particles and the description of the algorithmic framework used are given in Section 2. Section 3 presents the two leader election algorithms along with the associated correctness proofs. In Section 4, algorithms for assigning global and $\ell$-local identifiers are presented. Section 5 concludes the paper.

\section{Notation, definitions and our programmable matter algorithmic framework}

\subsection{FCC grid}
\label{sec:2.1}
The face-centered cubic grid, denoted by $\mathcal{F}$, is the graph with vertex set $\{(i,j,k)|\ i\in \mathbb{Z},\ j\in \mathbb{Z},\ k\in\mathbb{Z}_{2} \}\cup\{(i+0.5,j+0.5,k)|\ i\in \mathbb{Z},\ j\in \mathbb{Z},\ k\in\mathbb{Z}_{1} \}$, where $\mathbb{Z}_{2}$ is the set of the even integers and $\mathbb{Z}_{1}=\mathbb{Z}\setminus\mathbb{Z}_{2}$ is the set of odd integers, and edge set $\{(i,j,k) (i',j',k')|\ (|i-i'|=1\land j=j'\land k=k' )\lor (i=i'\land |j-j'|=1\land k=k')\lor(|i-i'|=1/2\land|j-j'|=1/2 \land |k-k'|=1) \}$. A subgraph of this grid is illustrated in Figure~\ref{fig:FCC}, showing the neighborhood of a vertex.

Note that by considering a triplet $(x,y,z)$ of $V(\mathcal{F})$ as the cardinal position of the vertex in a $3$-dimensional space, we obtain a representation of the whole grid in the $3$-dimensional space as in Figure \ref{fig:FCC}.

The layer $k$ of the grid $\mathcal{F}$ is the subset of vertices  $\{(i,j,k)|\ i\in \mathbb{Z},\ j\in \mathbb{Z} \}$, if $k$ is even or the subset of vertices $\{(i+0.5,j+0.5,k)|\ i\in \mathbb{Z},\ j\in \mathbb{Z} \}$, if $k$ is odd. Note that the graph induced by the vertices of layer $k$ is isomorphic to a square grid, see Figure~\ref{fig:FCC}. 
We denote by $\mathcal{F}_k$, the subgraph induced by the vertices of layer $k$. 

\begin{figure}
\centering
\tdplotsetmaincoords{82}{120}
\begin{tikzpicture}[tdplot_main_coords,scale=1.8]
      \draw[ultra thin,color=red] (2,2,0) -- (2.5,2.5,1.2) -- (2,3,0);
      \draw[ultra thin,color=red] (3,2,0) -- (2.5,2.5,1.2) -- (3,3,0);
      \draw[ultra thin,color=red] (2,2,2.2) -- (2.5,2.5,1.2) -- (2,3,2.2);
      \draw[ultra thin,color=red] (3,2,2.2) -- (2.5,2.5,1.2) -- (3,3,2.2);
\foreach \i in {0,...,4}
 \foreach \j in {0,...,4}
    {
      \ifthenelse{\i < 4}{\draw[thin,color=blue] (\i,\j,0) -- (\i+1,\j,0);}{} 
      \ifthenelse{\i < 4}{\draw[thin,color=blue] (\i,\j,2.2) -- (\i+1,\j,2.2);}{} 
      \ifthenelse{\i < 4}{\draw[thin,color=blue] (\i+.5,\j+.5,1.2) -- (\i+1.5,\j+.5,1.2);}{} 
      \ifthenelse{\j < 4}{\draw[thin,color=blue] (\i,\j,0) -- (\i,\j+1,0);}{} 
      \ifthenelse{\j < 4}{\draw[thin,color=blue] (\i,\j,2.2) -- (\i,\j+1,2.2);}{} 
      \ifthenelse{\j < 4}{\draw[thin,color=blue] (\i+.5,\j+.5,1.2) -- (\i+.5,\j+1.5,1.2);}{}
      \node[draw, circle, fill, color=gray, minimum size=1pt,inner sep=.7pt] at (\i,\j,0) (x\i\j){};
      \node[draw, circle, fill, color=gray, minimum size=1pt,inner sep=.7pt] at (\i+.5,\j+.5,1.2) (y\i\j){};
      \node[draw, circle, fill, color=gray, minimum size=1pt,inner sep=.7pt] at (\i,\j,2.2) (z\i\j){};
    }
  \draw[thin,color=red](1.5,2.5,1.2) -- (2.5,2.5,1.2) -- (3.5,2.5,1.2);
  \draw[thin,color=red] (2.5,1.5,1.2) -- (2.5,2.5,1.2) -- (2.5,3.5,1.2);
  \node[draw, circle, fill, color=red, minimum size=5pt,inner sep=1pt] at (2.5,2.5,1.2) (){};

\end{tikzpicture}
\caption{Neighborhood in the face-centered cubic grid $\mathcal{F}$}
\label{fig:FCC}
\end{figure}
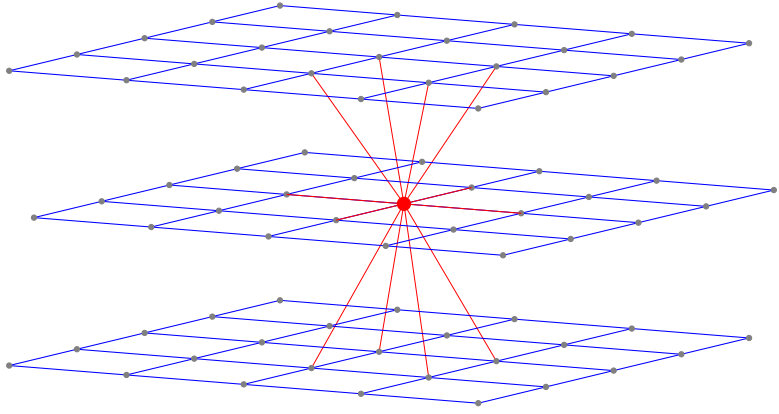

By hypothesis, any structure the different particles can form is a subgraph of $\mathcal{F}$.
In this graph, the vertex set $V(\mathcal{F})$ represents all the possible positions that the particles can occupy and the edge set $E(\mathcal{F})$ corresponds to the possible connections between particles (and consequently possible communications). The following two paragraphs present the notation and definitions we use for graphs.

We denote by $d_G(u,v)$, the usual distance between two vertices $u$ and $v$ in a graph $G$. If we consider the distance in a subgraph $H$ of $G$, the distance will be denoted by $d_H(u,v)$. 
We denote by $diam(G)$, the diameter of graph, i.e., the minimum integer $k$ such that any two vertices $u\in V(G)$ and $v\in  V(G)$ satisfy $d_G(u,v)\le k$.
The set $N_G(u)=\{ v\in V(G) |\ uv \in E(G)\}$ is the set of \textit{neighbors} of $u$.
Finally, we denote by $G[S]$, for $S\subseteq V(G)$, the subgaph induced by the vertices from $S$ and by $G-S$ the subgraph of $G$ induced by the vertices from $V(G)\setminus S$. In this paper, we use the notation $N_S(u)$, for $S$ a set of vertices of $\mathcal{F}$, to denote the set $N_{\mathcal{F}[S]}(u)$.

Let $p_{+}$ be the function such that $p_{+}(k)=k$ if $k\ge 0$ and $p_{+}(k)=0$ otherwise.
The distance between two vertices $(i,j,k)$ and $(i',j',k')$ of $\mathcal{F}$ is given by the following formula \cite{GT2019}:
$$d_{\mathcal{F}}((i,j,k),(i',j',k'))=p_{+}\left(|i-i'|-\frac{|k-k'|}{2}\right)+p_{+}\left(|j-j'|-\frac{|k-k'|}{2}\right)+ |k-k'|.$$

\subsection{Algorithmic framework}

The remaining part of this subsection is dedicated to the presentation of our programmable matter algorithmic framework.

\paragraph{Particles and ports.}
Our main assumption is that each particle occupies a single vertex of $\mathcal{F}$ and that each vertex is occupied by at most one particle.
The set of particles will be denoted by $S$ and, for sake of simplicity, we will interchangeably consider a particle in $S$ or the corresponding vertex in $\mathcal{F}$. Thus, it allows us to write $N_S(p)$ for the set of neighboring particles of a particle $p$ and $\mathcal{F}[S]$ for the subgraph induced by the set of particles $S$. In all the paper, this subgraph $\mathcal{F}[S]$ is supposed to be connected.
The \textit{ports} of a particle are the endpoints of communication. Each particle has 12 ports in $\mathcal{F}$, each labeled by a different integer from $\{0,\ldots,11\}$ (since each vertex of $\mathcal{F}$ has twelve neighbors). The ports of a particle occupying a vertex $u$ are represented by the edges incident with $u$. 
An edge between two vertices represents a possible communication between two particles $p_{1}$ and $p_{2}$ occupying these two vertices using each one a possibly different port. 
Particle have the following properties:
\begin{itemize}
\item each particle is anonymous, i.e., it does not have an identifier;
\item each particle knows the labels of the ports in contact with particles from its neighborhood;
\item each particle knows, for each pair of particles $q$ and $q'$ from its closed neighborhood, which ports of $q$ and $q'$ are in contact with a same particle.
\end{itemize} 

We recall that the closed neighborhood of $p$ is the neighborhood of $p$ in which we have added the particle $p$.
We consider that the above three properties are reasonable. Specifically, for the last property, we can assume that before the execution of the algorithm, for each pair of particles $q$ and $q'$ from its neighborhood, a particle $p$ can send a message containing the port labels of $p$ in contact with $q$ and $q'$ to $q$ and $q'$ which will be re-transmitted to the neighbors of $q$ and $q'$ and see if a particle receives two messages. If that is the case, this particle is a common neighbor of $q$ and $q'$. However, although the removing of this last property as hypothesis does not change the correctness of Algorithm \ref{alg2}, it can decrease the performance of Algorithm \ref{alg2} in term of required number of rounds to finish (see the definition at the end of the section). More precisely, the required number of rounds can be $t+1$ times more than without this hypothesis, $t$ being a constant representing the required number of rounds to compute, for a particle $p$ and for each pair of particles $q$ and $q'$ from its closed neighborhood, which ports of $q$ and $q'$ are in contact with a same particle.

In this paper, we consider two distinct hypothesis on the orientation of the particles that we call the heterogeneous and homogeneous cases (the choice of hypothesis will influence the way to do the leader election):

\begin{itemize}
\item In the {\em homogeneous case}, each two particles are labeled exactly in the same way, i.e., the ports $i$ of all particles are all in the same direction for all $i\in\{0,\ldots,11\}$.
\item In the {\em heterogeneous case}, for each two particles $p$ and $p'$ there exists a rotation function such that $p$ and $p'$ are labeled exactly in the same way.
\end{itemize} 

Note that the homogeneous case is a particular case of the heterogeneous case, i.e., any algorithm which solves the leader election problem in the heterogeneous case also solves the problem in the homogeneous case. 
We also remark that in  the other works~\cite{BAZ2019,JJD2017,ZD2015b}, the authors make assumptions (in the 2-dimensional space) using the terms of {\em common chirality} of a globally consistent circular orientation of the plane shared by all particles. In this paper, we consider that the particles of a layer do not have a common chirality (the ports of the particles are labeled following the clockwise order or the counter clockwise order).
We do this assumption in order to consider that a particle may have been turned upside down since its conception.

For a particle $p$, we denote by $N_S^{\pm}(p)$ the set of particles in $S$ which are in the neighborhood of $p$ but in a different layer of $\mathcal{F}$ than $p$. We use the notation $N_S^0(p)$ to denote the set of particles in $S$ which are in the neighborhood of $p$ but in the same layer. Hence, $N_S(p)=N_S^0(p)\cup N_S^{\pm}(p)$.
For a particle $p$ occupying the position $(0,0,0)$, the neighbors of $p$ form the set $\{(1,0,0)$ $,(-1,0,0)$ $,(0,1,0)$ $,(0,-1,0),$ $(0.5,0.5,1),$ $(-0.5,0.5,1),$ $(0.5,-0.5,1),$ $(-0.5,-0.5,1),$ $(0.5,0.5,-1),$ $(-0.5,0.5,-1),$ $(0.5,-0.5,-1),$ $(-0.5,$ $-0.5,-1)\}$. 
Note that no matter which rotation from $N_{\mathcal{F}}(p)$ to $N_{\mathcal{F}}(p)$ has been done on $p$, a port which has been initially labeled in the direction of $(1,0,0)$ $,(-1,0,0)$ $,(0,1,0)$ or $(0,-1,0)$ will be in the direction of either $(1,0,0)$ $,(-1,0,0)$ $,(0,1,0)$ or $(0,-1,0)$. 
This is a consequence of the fact that only one plane contains exactly four neighbor of $p$ (this plane is the one containing the particles of the layer $0$).
Due to this fact, in our algorithmic framework, we suppose the following, that is only important in the heterogeneous case:

\begin{itemize}
\item each particle $p$ knows the labels of the ports by which it can communicate with particles in $N_S^{0}(p)$; and we assume without loss of generality that these ports form the set $\{0,1,2,3\}$ and that these port numbers are consecutive following the clockwise or counter-clockwise order around the particle;
\item each particle $p$ knows the labels of the ports by which it can communicate with particles in $N_S^{\pm}(p)$; these ports form the sets $\{4,5,6,7\}$ and $\{8,9,10,11\}$, with ports in $\{4,5,6,7\}$ allowing to communicate with particles in one of the two layers (below or above) and ports in $\{8,9,10,11\}$ allowing to communicate with particles in the other layer. In other words, $p$ knows if two particles in $N_S^{\pm}(p)$ are in the same layer or not. Moreover, it is supposed that ports $4$, 5, 6, 7 are in the opposite direction with ports $10$, 11, 8, 9, respectively and that both ports $4$, 5, 6, 7 and ports $10$, 11, 8, 9  are consecutive following the clockwise or counter-clockwise order around the particle. 

\end{itemize} 

Note that, as it is supposed that the particles are manufactured the same way (with initially the same numbering of ports), then the above assumptions are natural.

\paragraph{Computation model.}

The proposed algorithms in our algorithmic framework are results of successive local computations \cite{BAU2002,ROS1972}.
In particular, the leader election algorithm in the heterogeneous case can be described by a graph relabeling system \cite{BAU2002} which is a local computation system.

We suppose the following:

\begin{itemize}
\item each particle contains the same program and begins in the same state;
\item the computation process is represented by successive local computations;
\item no local computation occurs simultaneously on two particles at distance at most 2;
\item during a local computation, a particle can perform a bounded number of computations and can send messages to its neighbors;
\item a \emph{round} is a sequence of successive local computations for which each particle does at least one local  computation;
\item an algorithm finishes in $k$ rounds if after any $k$ successive rounds the algorithm is finished.
\end{itemize}
We moreover suppose that the particles act asynchronously, hence with the possibility that some particles act simultaneously.
Note that the concept of rounds is used to bound the running time of the algorithms. In general, it is possible to avoid that two particles at distance at most 2 do local computation simultaneously by using a probabilistic leader election algorithm on the vertices at distance at most $2$ of one of the two vertices, i.e., by computing a random value on the vertices at distance at most $2$ and doing the local computation following the increasing order of the values.

\section{Leader election in the face-centered cubic grid}
We propose in this section two algorithms (one for the homogeneous case and one for the heterogeneous case) for leader election, i.e., starting from all particles in state {\bf C} (Candidate), at the end of the execution of the algorithm, only one particle will be in state {\bf L} (Leader) and all the other will be in state {\bf N} (Not leader).

\subsection{Leader election in the homogeneous case}\label{homo}
In a first part of this subsection, the hypothesis made about the positions and the ports of the particles and the definitions that will be used in this subsection are given. In a second part, we give details about the states and the behavior of our algorithm.
The last part of this subsection is dedicated to the proof of the correctness of our algorithm. Finally, results about the required number of rounds and the space-complexity are given.

\subsubsection{Hypothesis about positions and definitions}
In the homogeneous case, we assume that all particles have their ports labeled in the same way. Hence, we can assume without loss of generality that any particle $p$ lying on vertex $(i,j,k)$ is connected through its ports $0,1,2,3,4,5,6,7,8,9,10,11$ to the twelve neighbors $(i-1,j,k)$, $(i,j+1,k)$, $(i+1,j,k)$, $(i,j-1,k),$ $(i-0.5,j+0.5,k-1),$ $(i+0.5,j+0.5,k-1),$ $(i+0.5,j-0.5,k-1),$ $(i-0.5,j-0.5,k-1),$ $(i-0.5,j+0.5,k+1),$ $(i+0.5,j+0.5,k+1),$ $(i+0.5,j-0.5,k+1),$ $(i-0.5,j-0.5,k+1)$, respectively. 
In other words, $p$ is able to know the coordinates of any of its neighbors through any port: the coordinates of the particle to which a particle $p$ lying on vertex $(i,j,k)$ is connected through port $a$ is given by $(I(i,a),J(j,a),K(k,a))$, where $I$, $J$ and $K$ are defined by:

$$I(i,a)=\left\{\begin{array}{ll}
		 i-1 &\text{ if } a=0\\
		 i-0.5 &\text{ if } a=4, 7, 8\text{ or } 11\\
                 i & \text{ if } a=1 \text{ or } 3\\
		 i+0.5 &\text{ if } a=5, 6, 9\text{ or } 10\\
		 i+1 &\text{ if } a=2
                \end{array}\right.$$

$$J(j,a)=\left\{\begin{array}{ll}
		 j-1 &\text{ if } a=3\\
		 j-0.5 &\text{ if } a=6, 7, 10\text{ or } 11\\
                 j & \text{ if } a=0 \text{ or } 2\\
		 j+0.5 &\text{ if } a=4, 5, 8\text{ or } 9\\
		 j+1 &\text{ if } a=1
                \end{array}\right.$$

$$K(k,a)=\left\{\begin{array}{ll}
		 k-1 &\text{ if } a=4, 5, 6 \text{ or } 7\\
                 k & \text{ if } a=0, 1, 2 \text{ or } 3\\
		 k+1 &\text{ if } a=8, 9, 10 \text{ or } 11
                \end{array}\right.$$

For two lists of triplets $L$ and $J$, we define the condition $c(L,J)$, by $c(L,J)=1$ if and only if for every triplet $(x,y,z)\in L$ and every triplet $(x',y',z')$ such that $(x,y,z)(x',y',z')\in \mathcal{F}$,  $(x',y',z')\in L\cup J$.
This condition will be used to check if the list $O(p)$, for a particle $p$, contains the relative position of every other particle of the system or not. More details about this condition are given in the next two parts.
Finally, let $\max(L)$, for $L$ a list of triplets, be the largest triplet of $L$ in the lexicographical order.
We recall that $(i,j,k)>(i',j',k')$, in the lexicographical order if $i>i'$ or $i=i'$ and $j>j'$ or $i=i'$, $j=j'$ and $k>k'$. 
In algorithm \ref{alg1}, by $>$, we mean that a list is larger than another list using the lexicographical order.

\subsubsection{Details about Algorithm \ref{alg1}}

\begin{algorithm}

\begin{algorithmic} 
\State \textbf{Case 1: } State \textbf{C}\\
For each port $a$ of $p$:
 \If {$a$ is in contact with a particle}
	\State send the message $(I(0,a),J(0,a),K(0,a),0)$ through port $a$. 
  \Else
	\State send the message $(I(0,a),J(0,a),K(0,a),1)$ through all ports $b$ in contact with particles. 
\EndIf
\State Set the state to \textbf{Lis}

\State \textbf{Case 2: } State \textbf{Lis}

\If {message $(i,j,k,\ell)$} is received through port $a$
    \If {$(i,j,k)\notin O(p)$ and $\ell=0$}
        \State add $(i,j,k)$ to $O(p)$
        \State{for each port $a$ of $p$, send the message $(I(i,a),J(j,a),K(k,a),0)$ through port $a$}
    \EndIf
    \If {$(i,j,k)\notin U(p)$ and $\ell=1$}
        \State add $(i,j,k)$ to $U(p)$
        \State{for each port $a$ of $p$, send the message $(I(i,a),J(j,a),K(k,a),1)$ through port $a$}
    \EndIf
\EndIf
\If {$c(O(p),U(p))=1$}
\State set the state to \textbf{I}
\EndIf
\State \textbf{Case 3: } State \textbf{I} \\
For each port $a$, send the message $(\max(O(p)),I(0,a),J(0,a),K(0,a))$ through port $a$
\State Set the state to \textbf{Mcomp}
\State \textbf{Case 4: } State \textbf{Mcomp}
\State Set $i=0$
\If {message $(m,(i,j,k))$ is received through port $a$ and $(i,j,k)\notin O'(p)$}
    \State Add $(i,j,k)$ to $O'(p)$
    \If {$(m>\max(O(p)) $}
    \State set the state to \textbf{N} and send the message $(m,I(i,a),J(j,a),K(k,a))$ through port $a$
    \EndIf
    \Else
        \If { $|O(p)|=|O'(p)|$} 
    		\State set the state to \textbf{L}
        \EndIf   
\EndIf
\State \textbf{Case 5: } States \textbf{N} and \textbf{L} \\
Perform no further actions
\end{algorithmic}
\caption{\label{alg1}The leader election in the homogeneous case, for particle $p$.} 
\end{algorithm}

In algorithm \ref{alg1}, we suppose that $O(p)$ ($O$ refers to occupied), $U(p)$ ($U$ refers to unoccupied) and $O'(p)$ are three lists of triplets being computed by each particle $p$. These three lists are pairwise different, for each pair of particles. However, the lists $O(p)$, $U(p)$ and $O'(p)$ are constructed in such a way that, for every particle $p$, they represent the three same (global) lists translated by a different vector.

Basically, Algorithm~\ref{alg1} consists in two main phases. In the first phase, each particle $p$ computes a list $O(p)$ of the positions of all the particles relatively to itself and in the second phase, the particle $p$ which has the largest $\max(O(p))$ is elected. For example, if there are only three particles $p_1$, $p_2$ and $p_3$, then for a possible configuration of the particles, the three computed lists can be $O(p_1)=\{(1,0,0), (2,0,0)\}$, $O(p_2)=\{(-1,0,0),(1,0,0)\}$ and $O(p_3)=\{(-2,0,0),(-1,0,0)\}$. The maximums of the three lists are then $(2,0,0)$, $(1,0,0)$ and $(-1,0,0)$, respectively. The elected particle will be $p_1$ (since $(2,0,0)$ is larger than both $(1,0,0)$ and $(-1,0,0)$).

The intermediate states defined in Algorithm~\ref{alg2} are: \textbf{Lis} (List construction) that corresponds to the state where the lists $O(p)$ and $U(p)$, for a particle $p$, are constructed; the state \textbf{I} (Initialization) when an initial message (containing $\max(O(p))$) is transmitted to all the other particles and the state \textbf{Mcomp} (Maximum comparison) in which a particle $p$ checks if $\max(O(p))$ is larger than every other $\max$. Before the execution of the algorithm both $O(p)$, $U(p)$ and $O'(p)$ are supposed to be empty and each particle is in state \textbf{C} (Candidate).

In the first phase (Cases 1 and 2), each particle $p$ computes two lists $O(p)$ and $U(p)$ containing the local positions of the particles of $S$ and $U(p)$ containing the local positions of the empty vertices adjacent to a particle of $S$.
If a particle $p$ is in state \textbf{C}, \textbf{Lis} or \textbf{I}, we suppose that the messages received from particles in state \textbf{Mcomp} are buffered until $p$ is in state \textbf{Mcomp} and at this moment $p$ deals with these buffered messages.
In the list $O(p)$, the position of $p$ is supposed to be $(0,0,0)$ and, for example, a possible existing neighbor connected through the port 0 of $p$ will corresponds to the triplet $(1,0,0)$ in the list $O(p)$. 
This phase stops for a particle when the condition $c(O(p),U(p))$ is equal to $1$, i.e., when we are assured that any particle at a local position in $O(p)$ (relatively to $p$) has all its neighbors (occupied or not) either in $O(p)$ or $U(p)$.
This condition allows, for a particle $p$, to be aware if the list $O(p)$ contains the local position of every other particle of the system or not and, thus, if it still has to wait for a not yet received message or not.

In the second phase (Cases 3, 4 and 5), each particle $p$ computes the $\max$ of $O(p)$ and compares it with the value calculated by the other particles. In order to know when a particle has compared its $\max$ with every other $\max$ computed by the other particles, each particle $p$ constructs a list $O'(p)$ containing the position of the particles for which it has already did the inequality test. We affect the state \textbf{L} (Leader) to the particle $p$ for which $\max(O(p))$ is larger than every other $\max$. In order to check this, we add a value to $O'(p)$ each time $\max(O(p))$ is larger than an other $\max$. When $O'(p)$ is as large as $O(p)$, we affect the state \textbf{L} to the particle $p$.

\subsubsection{Correctness, required number of rounds and space-complexity}

\begin{prop}
At the end of the execution of Algorithm \ref{alg1}, exactly one particle is in state \textbf{L}.
\end{prop}
\begin{proof}
Since the lexicographical order is a total order, there exists at least a particle $p$ such that $\max(O(p))$ is larger or equal than the value $\max(O(p'))$, for any other particle $p'$.
Thus, there is at least one particle in state \textbf{L} at the end of the execution of Algorithm \ref{alg1}. It remains to prove that there are no more than one particle in state \textbf{L}.

Suppose, by contradiction, that two distinct particles $p_1$ and $p_2$ are in state $\textbf{L}$. It implies that $\max(O(p_1))=\max(O(p_2))$. 
Let $(i_1,j_1,k_1)$ be the vertex occupied by $p_1$ and let $(i_2,j_2,k_2)$ be the vertex occupied by $p_2$.
By construction there exists a particle $p'_1$ at vertex $(i'_1,j'_1,k'_1)$  such that $\max(O(p_1))=(i'_1-i_1,j'_1-j_1,k'_1-k_1)$ and a particle $p'_2$ at vertex $(i'_2,j'_2,k'_2)$  such that $\max(O(p_1))=(i'_2-i_2,j'_2-j_2,k'_2-k_2)$. Since $p_1$ and $p_2$ are distinct, either $(i_1,j_1,k_1)<(i_2,j_2,k_2)$ or $(i_1,j_1,k_1)>(i_2,j_2,k_2)$, contradicting the fact that both $(i'_1-i_2,j'_1-j_2,k'_1-k_2)<(i'_1-i_1,j'_1-j_1,k'_1-k_1)$ and $(i'_2-i_1,j'_2-j_1,k'_2-k_1)<(i'_2-i_2,j'_2-j_2,k'_2-k_2)$.

Note that a particle $p$ such that $\max(O(p))$ is larger than the value $\max(O(p'))$, for any other particle $p'$, is aware of this property (and, consequently, become leader) by adding the position of each particle $p'$ such that $\max(O(p'))<\max(O(p))$ in a list $O'(p)$ and by checking if $O'(p)$ contains the positions of all the other particles, i.e., by checking if $|O(p)|=|O'(p)|$.
\end{proof}
Note that, by definition of round and distance, and by the fact that each particle receiving a message retransmits it on all its ports, we are sure that after $d$ rounds a particle $p$ should have received all messages initially emitted by particles at distance $d$ of $p$.

In Algorithm \ref{alg1}, it can be noted that each particle $p$ has to check if $c(O(p),U(p))=1$ and has to compute the value $\max(O(p))$. Thus, in order that Algorithm \ref{alg1} works properly, it requires that each particle is able to compute $O(|S|^2)$ operations (as the size of $O(p)$ and $U(p)$ is $O(|S|)$ ). The required number of operations to compute can be decreased by implementing an incremental algorithm to check if $c(O(p),U(p))=1$.
 Moreover, about the required memory space and required number of rounds in order to guarantee the end of the execution, we state the following. 
\begin{prop}
Whatever the structure of $\mathcal{F}[S]$, if $\mathcal{F}[S]$ is connected, then a particle will be in state \textbf{L} after (at most) $2 diam(\mathcal{F}[S])+2$ rounds since the beginning of the execution of Algorithm \ref{alg1}. Moreover, for each particle, the memory space used by Algorithm \ref{alg1} is bounded by $O(|S| \log(|S|))$ and the size of the messages is bounded by  $O(\log(|S|))$.
\end{prop}
\begin{proof}
We begin by proving that after $2 diam(\mathcal{F} [S])+2$ rounds a particle is in state \textbf{L}. We recall that at the beginning of the execution, it is supposed that all particles are in state \textbf{C}. First, after one round no more particle is in state \textbf{C}. Second, after $diam(O(p))$ rounds no particle remains in state \textbf{Lis}. Third, after one round no more particle is in state \textbf{Comp}. Finally, after $diam(O(p))$ we are sure that $|O(p)|=|O'(p)|$ (if its own state has not been set to \textbf{N}), for any particle $p$. Consequently, after $2 diam(O(p))+2$ rounds a particle is in state \textbf{L}.

Note that for any particle $p$ and any $(x,y,z)\in O(p)\cup U(p)$, both $|x|$, $|y|$ and $|z|$ are bounded by $|S|+1$. This implies that the size of the messages is bounded by $O(\log(|S|))$. Since, both $O(p)$, $U(p)$ and $O(p')$ have at most $12|S|$ elements, the memory space used by Algorithm \ref{alg1} is bounded by $O(|S| \log(|S|))$. 
\end{proof}
\subsection{Leader election in the heterogeneous case}\label{hetero}

In the heterogeneous case, we cannot assume that all particles have their ports labeled in the same directions, hence Algorithm~\ref{alg1} will not always produce a leader. 
In this subsection, our motivation is to give a leader election algorithm (Algorithm~\ref{alg2}) that works in the heterogeneous case. However, since our hypothesis are weaker that in the homogeneous case, finding an algorithm working for every configuration of particles (if it exists) is more challenging. Consequently, we give an algorithm where the correctness is proven for specific configurations of particles. In contrast  with the leader election algorithm that works in the homogeneous case, the presented algorithm is easier to implement and only needs a constant size memory per particle.

The key idea of the Algorithm~\ref{alg2} is, starting from the set $S$ of particles in state \textbf{C} (Candidate), to let the particles at the 'periphery' of $S$  to change their state to \textbf{N} (Not elected) until there remains only one particle that will change its state to \textbf{L} (Leader) and will become the leader. For this, we extend the definition of contractability that was used in our proposed algorithm concerning the 2-dimensional case~\cite{GT2018}.

This subsection is organized in four parts. A first part concerns the definition of a contractible particle, the main concept used in Algorithm~\ref{alg2}. This part also contains a proposition about the relation between isometric graphs and contractible particles. The second part concerns the electable sets, i.e., some configurations of particles for which we are sure that Algorithm~\ref{alg2} behaves correctly. A proof about the correctness of Algorithm~\ref{alg2} is given for electable sets. A third part is dedicated to explain in some details how Algorithm~\ref{alg2} works.
The last part is about the required number of rounds and space-complexity of Algorithm~\ref{alg2}.
Also, in this part, we prove the existence of a polynomial time algorithm to check if a set is electable.

\subsubsection{Contractibility and isometricity}
For a particle $p$ occupying a vertex $(i,j,k)$ of $\mathcal{F}$, the four vertices $(i+1,j+1,k)$, $(i-1,j+1,k)$, $(i+1,j-1,k)$ and $(i-1,j-1,k)$ are the \textit{corners} of $p$ (vertices at distance 2 from $p$ on the same layer) and the set of corners is denoted by $C(p)$. The \emph{extended neighborhood} $M_S(p)$ of a particle $p$ is defined by $M_S(p)=N_S^{0}(p)\cup C(p)\cap S$.

\begin{de}\label{defcontractible}

For a set $S\subseteq V(\mathcal{F})$ of particles, a particle $p$ in $S$ is said to be \emph{$S$-contractible} if it satisfies the following three properties:
\begin{enumerate}
\item[I)] $\mathcal{F}[M_S(p)]$ is connected;
\item[II)] $|N_S^0(p)|\le 2$;
\item[III)] $\mathcal{F}[N_S^{\pm}(p)]$ is connected and either $N_S^0(p)=\emptyset$ or $N_S^{\pm}(p)=\emptyset$ or there exits a particle $p'\in N_S^0(p)$ such that $N_S^{\pm}(p) \cap N_S^{\pm}(p')\neq\emptyset$.
\end{enumerate}
\end{de}

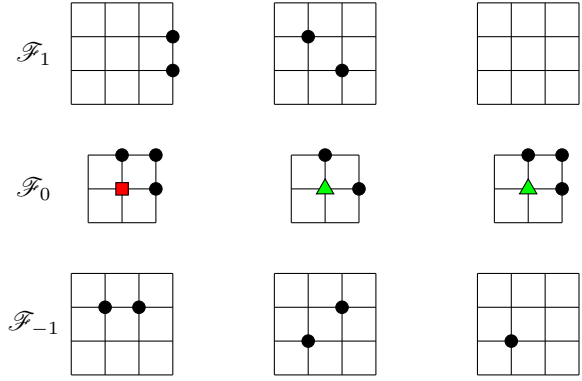
\begin{figure}[t]
\begin{center}
\begin{tikzpicture}[scale=0.90]

\draw (-0.5-3,0+2) -- (0.5-3,0+2);
\draw (-0.5-3,0.5+2) -- (0.5-3,0.5+2);
\draw (-0.5-3,-0.5+2) -- (0.5-3,-0.5+2);
\draw (0-3,-0.5+2) -- (0-3,0.5+2);
\draw (0.5-3,-0.5+2) -- (0.5-3,0.5+2);
\draw (-0.5-3,-0.5+2) -- (-0.5-3,0.5+2);

\node at (0-3,0+2)[regular polygon, regular polygon sides=4,draw=black,fill=red,scale=0.5]{};
\node at (0.5-3,0+2)[circle,draw=black,fill=black,scale=0.5]{};
\node at (0-3,0.5+2)[circle,draw=black,fill=black,scale=0.5]{};
\node at (0.5-3,0.5+2)[circle,draw=black,fill=black,scale=0.5]{};

\node at (-4.3,2){$\mathcal{F}_0$};
\node at (-4.3,0){$\mathcal{F}_{-1}$};
\node at (-4.3,4){$\mathcal{F}_{1}$};

\draw (-0.75-3,0.25) -- (0.75-3,0.25);
\draw (-0.75-3,0.75) -- (0.75-3,0.75);
\draw (-0.75-3,-0.25) -- (0.75-3,-0.25);
\draw (-0.75-3,-0.75) -- (0.75-3,-0.75);

\draw (0.25-3,-0.75) -- (0.25-3,0.75);
\draw (0.75-3,-0.75) -- (0.75-3,0.75);
\draw (-0.25-3,-0.75) -- (-0.25-3,0.75);
\draw (-0.75-3,-0.75) -- (-0.75-3,0.75);

\node at (0.25-3,0.25)[circle,draw=black,fill=black,scale=0.5]{};
\node at (-0.25-3,0.25)[circle,draw=black,fill=black,scale=0.5]{};

\draw (-0.75-3,0.25+4) -- (0.75-3,0.25+4);
\draw (-0.75-3,0.75+4) -- (0.75-3,0.75+4);
\draw (-0.75-3,-0.25+4) -- (0.75-3,-0.25+4);
\draw (-0.75-3,-0.75+4) -- (0.75-3,-0.75+4);

\draw (0.25-3,-0.75+4) -- (0.25-3,0.75+4);
\draw (0.75-3,-0.75+4) -- (0.75-3,0.75+4);
\draw (-0.25-3,-0.75+4) -- (-0.25-3,0.75+4);
\draw (-0.75-3,-0.75+4) -- (-0.75-3,0.75+4);

\node at (0.75-3,0.25+4)[circle,draw=black,fill=black,scale=0.5]{};
\node at (0.75-3,-0.25+4)[circle,draw=black,fill=black,scale=0.5]{};

\draw (-0.5,0+2) -- (0.5,0+2);
\draw (-0.5,0.5+2) -- (0.5,0.5+2);
\draw (-0.5,-0.5+2) -- (0.5,-0.5+2);
\draw (0,-0.5+2) -- (0,0.5+2);
\draw (0.5,-0.5+2) -- (0.5,0.5+2);
\draw (-0.5,-0.5+2) -- (-0.5,0.5+2);

\node at (0,0+2)[regular polygon, regular polygon sides=3,draw=black,fill=green,scale=0.4]{};
\node at (0.5,0+2)[circle,draw=black,fill=black,scale=0.5]{};
\node at (0,0.5+2)[circle,draw=black,fill=black,scale=0.5]{};

\draw (-0.75,0.25) -- (0.75,0.25);
\draw (-0.75,0.75) -- (0.75,0.75);
\draw (-0.75,-0.25) -- (0.75,-0.25);
\draw (-0.75,-0.75) -- (0.75,-0.75);
\draw (0.25,-0.75) -- (0.25,0.75);
\draw (0.75,-0.75) -- (0.75,0.75);
\draw (-0.25,-0.75) -- (-0.25,0.75);
\draw (-0.75,-0.75) -- (-0.75,0.75);

\node at (0.25,0.25)[circle,draw=black,fill=black,scale=0.5]{};
\node at (-0.25,-0.25)[circle,draw=black,fill=black,scale=0.5]{};

\draw (-0.75,0.25+4) -- (0.75,0.25+4);
\draw (-0.75,0.75+4) -- (0.75,0.75+4);
\draw (-0.75,-0.25+4) -- (0.75,-0.25+4);
\draw (-0.75,-0.75+4) -- (0.75,-0.75+4);
\draw (0.25,-0.75+4) -- (0.25,0.75+4);
\draw (0.75,-0.75+4) -- (0.75,0.75+4);
\draw (-0.25,-0.75+4) -- (-0.25,0.75+4);
\draw (-0.75,-0.75+4) -- (-0.75,0.75+4);

\node at (-0.25,0.25+4)[circle,draw=black,fill=black,scale=0.5]{};
\node at (0.25,-0.25+4)[circle,draw=black,fill=black,scale=0.5]{};

\draw (-0.5+3,0+2) -- (0.5+3,0+2);
\draw (-0.5+3,0.5+2) -- (0.5+3,0.5+2);
\draw (-0.5+3,-0.5+2) -- (0.5+3,-0.5+2);
\draw (0+3,-0.5+2) -- (0+3,0.5+2);
\draw (0.5+3,-0.5+2) -- (0.5+3,0.5+2);
\draw (-0.5+3,-0.5+2) -- (-0.5+3,0.5+2);

\node at (0+3,0+2)[regular polygon, regular polygon sides=3,draw=black,fill=green,scale=0.4]{};
\node at (0.5+3,0+2)[circle,draw=black,fill=black,scale=0.5]{};
\node at (0+3,0.5+2)[circle,draw=black,fill=black,scale=0.5]{};
\node at (0.5+3,0.5+2)[circle,draw=black,fill=black,scale=0.5]{};

\draw (-0.75+3,0.25) -- (0.75+3,0.25);
\draw (-0.75+3,0.75) -- (0.75+3,0.75);
\draw (-0.75+3,-0.25) -- (0.75+3,-0.25);
\draw (-0.75+3,-0.75) -- (0.75+3,-0.75);
\draw (0.25+3,-0.75) -- (0.25+3,0.75);
\draw (0.75+3,-0.75) -- (0.75+3,0.75);
\draw (-0.25+3,-0.75) -- (-0.25+3,0.75);
\draw (-0.75+3,-0.75) -- (-0.75+3,0.75);

\node at (-0.25+3,-0.25)[circle,draw=black,fill=black,scale=0.5]{};

\draw (-0.75+3,0.25+4) -- (0.75+3,0.25+4);
\draw (-0.75+3,0.75+4) -- (0.75+3,0.75+4);
\draw (-0.75+3,-0.25+4) -- (0.75+3,-0.25+4);
\draw (-0.75+3,-0.75+4) -- (0.75+3,-0.75+4);
\draw (0.25+3,-0.75+4) -- (0.25+3,0.75+4);
\draw (0.75+3,-0.75+4) -- (0.75+3,0.75+4);
\draw (-0.25+3,-0.75+4) -- (-0.25+3,0.75+4);
\draw (-0.75+3,-0.75+4) -- (-0.75+3,0.75+4);

\end{tikzpicture}
\end{center}
\caption{A $S$-contractible particle $p$ (at the left) and two non $S$-contractible particles $p'$ and $p''$ (at the middle and right), both of them being in the layer 0 of $\mathcal{F}$ (square: $p$; triangle: $p'$ or $p''$; circle: other particle of $S$).}
\label{twoparticles}
\end{figure}

Remark that $\mathcal{F}[N_S^{\pm}(p)]$ is connected, for a particle $p$, (this property appears in condition III) of the above definition) implies that $p$ does not have neighbors in both the layer below and above.

The left part of Figure \ref{twoparticles} illustrates an $S$-contractible particle of $\mathcal{F}$ satisfying conditions I), II) and III). In contrast with the left part, the middle part of Figure \ref{twoparticles} illustrates a non $S$-contractible particle of $\mathcal{F}$ satisfying condition II), but which does not satisfy conditions I) and III) (since $\mathcal{F}[N_S^{\pm}(p')]$ is not connected) and the right part of Figure \ref{twoparticles} illustrates a non $S$-contractible particle of $\mathcal{F}$ satisfying conditions I) and II), but which does not satisfy condition III) (since there does not exist a particle $q\in N_S^0(p'')$ such that $N_S^{\pm}(p'') \cap N_S^{\pm}(q)\neq\emptyset$ ). In this figure (and also in Figures \ref{twosets}, \ref{examplexec} and \ref{portreconfig}), it is supposed that the neighbors of a vertex $q$ in $\mathcal{F}$ from Layer 0 are represented, in Layer $1$, by the four vertices at the closest distance from the geographical position that $q$ would have in this layer and the same goes for Layer $-1$. For example, the four neighbors of the particle $p$ in Layer $-1$ of $\mathcal{F}$ (in the left part of Figure \ref{twoparticles}) are the two vertices occupied by particles of $S$ and the two unoccupied vertices which are in the central square of this figure.

A particle $p$ can detect if it is $S$-contractible or not by checking if it satisfies conditions I), II) and III) as follows. Particle $p$ can test if it satisfies conditions I) and II) by checking if there are at most two different occupied ports from $N_S^{0}(p)$ and by checking, in the case there are two, that their labels are successive and that they have a common neighbor (except $p$). By hypothesis, a particle can verify if its two neighbors have a common neighbor without computation.
Moreover, $p$ can test if $\mathcal{F}[N_S^{\pm}(p)]$ is connected (see condition III)), by first checking if there are particles of $N_S^{\pm}(p)$ in at most one layer, and in the case $|N_S^{\pm}(p)|=2$, it checks if the two particles from this set are adjacent by checking if their labels are successive. Also, by hypothesis, $p$ can easily check if $N_S^{\pm}(p)=\emptyset$. Finally, $p$ can check if there exits a particle $p'\in N_S^0(p)$ such that $N_S^{\pm}(p) \cap N_S^{\pm}(p')\neq\emptyset$ by checking for each particle $q$ of $N_S^{\pm}(p)$, if there is (at least) a common neighbor between $p$ and $q$ in $N_S^0(p)$.

For a given integer $i$, a subgraph $G$ of $\mathcal{F}_i$ is said to be {\em isometric} if for any two vertices $u,v$ of $G$ we have $d_{G}(u,v)=d_{\mathcal{F}_i}(u,v)$. Note that, by definition, an isometric subgraph of $\mathcal{F}_i$ is connected.

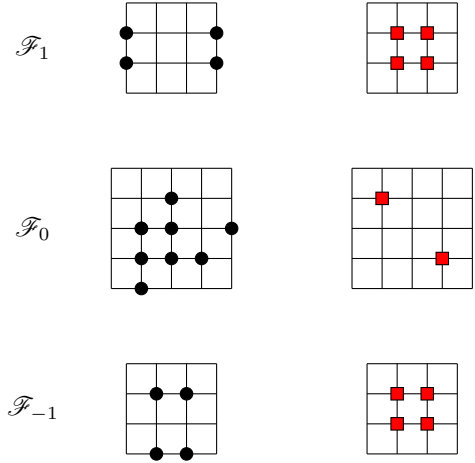
\begin{figure}[t]
\begin{center}
\begin{tikzpicture}[scale=0.80]

\draw (-1-4,1+3) -- (1-4,1+3);
\draw (-1-4,0+3) -- (1-4,0+3);
\draw (-1-4,0.5+3) -- (1-4,0.5+3);
\draw (-1-4,-0.5+3) -- (1-4,-0.5+3);
\draw (-1-4,-1+3) -- (1-4,-1+3);
\draw (0-4,-1+3) -- (0-4,1+3);
\draw (0.5-4,-1+3) -- (0.5-4,1+3);
\draw (-0.5-4,-1+3) -- (-0.5-4,1+3);
\draw (1-4,-1+3) -- (1-4,1+3);
\draw (-1-4,-1+3) -- (-1-4,1+3);

\node at (0-4.5,0+3)[circle,draw=black,fill=black,scale=0.5]{};
\node at (0-4,0.5+3)[circle,draw=black,fill=black,scale=0.5]{};
\node at (1-4,0+3)[circle,draw=black,fill=black,scale=0.5]{};
\node at (-0.5-4,0+3)[circle,draw=black,fill=black,scale=0.5]{};
\node at (0-4,0+3)[circle,draw=black,fill=black,scale=0.5]{};
\node at (0.5-4,-0.5+3)[circle,draw=black,fill=black,scale=0.5]{};
\node at (-0.5-4,-0.5+3)[circle,draw=black,fill=black,scale=0.5]{};
\node at (0-4,-0.5+3)[circle,draw=black,fill=black,scale=0.5]{};
\node at (-0.5-4,-1+3)[circle,draw=black,fill=black,scale=0.5]{};

\node at (-6.3,3){$\mathcal{F}_0$};
\node at (-6.3,0){$\mathcal{F}_{-1}$};
\node at (-6.3,6){$\mathcal{F}_{1}$};

\draw (-0.75-4,0.25) -- (0.75-4,0.25);
\draw (-0.75-4,0.75) -- (0.75-4,0.75);
\draw (-0.75-4,-0.25) -- (0.75-4,-0.25);
\draw (-0.75-4,-0.75) -- (0.75-4,-0.75);
\draw (0.25-4,-0.75) -- (0.25-4,0.75);
\draw (0.75-4,-0.75) -- (0.75-4,0.75);
\draw (-0.25-4,-0.75) -- (-0.25-4,0.75);
\draw (-0.75-4,-0.75) -- (-0.75-4,0.75);

\node at (0.25-4,0.25)[circle,draw=black,fill=black,scale=0.5]{};
\node at (-0.25-4,0.25)[circle,draw=black,fill=black,scale=0.5]{};
\node at (0.25-4,-0.75)[circle,draw=black,fill=black,scale=0.5]{};
\node at (-0.25-4,-0.75)[circle,draw=black,fill=black,scale=0.5]{};

\draw (-0.75-4,0.25+6) -- (0.75-4,0.25+6);
\draw (-0.75-4,0.75+6) -- (0.75-4,0.75+6);
\draw (-0.75-4,-0.25+6) -- (0.75-4,-0.25+6);
\draw (-0.75-4,-0.75+6) -- (0.75-4,-0.75+6);
\draw (0.25-4,-0.75+6) -- (0.25-4,0.75+6);
\draw (0.75-4,-0.75+6) -- (0.75-4,0.75+6);
\draw (-0.25-4,-0.75+6) -- (-0.25-4,0.75+6);
\draw (-0.75-4,-0.75+6) -- (-0.75-4,0.75+6);

\node at (0.75-4,0.25+6)[circle,draw=black,fill=black,scale=0.5]{};
\node at (0.75-4,-0.25+6)[circle,draw=black,fill=black,scale=0.5]{};
\node at (-0.75-4,0.25+6)[circle,draw=black,fill=black,scale=0.5]{};
\node at (-0.75-4,-0.25+6)[circle,draw=black,fill=black,scale=0.5]{};

\draw (-1,1+3) -- (1,1+3);
\draw (-1,0+3) -- (1,0+3);
\draw (-1,0.5+3) -- (1,0.5+3);
\draw (-1,-0.5+3) -- (1,-0.5+3);
\draw (-1,-1+3) -- (1,-1+3);
\draw (0,-1+3) -- (0,1+3);
\draw (0.5,-1+3) -- (0.5,1+3);
\draw (-0.5,-1+3) -- (-0.5,1+3);
\draw (1,-1+3) -- (1,1+3);
\draw (-1,-1+3) -- (-1,1+3);

\node at (0.5,0+2.5)[regular polygon, regular polygon sides=4,draw=black,fill=red,scale=0.5]{};
\node at (-0.5,1+2.5)[regular polygon, regular polygon sides=4,draw=black,fill=red,scale=0.5]{};

\draw (-0.75,0.25) -- (0.75,0.25);
\draw (-0.75,0.75) -- (0.75,0.75);
\draw (-0.75,-0.25) -- (0.75,-0.25);
\draw (-0.75,-0.75) -- (0.75,-0.75);
\draw (0.25,-0.75) -- (0.25,0.75);
\draw (0.75,-0.75) -- (0.75,0.75);
\draw (-0.25,-0.75) -- (-0.25,0.75);
\draw (-0.75,-0.75) -- (-0.75,0.75);

\node at (0.25,0.25)[regular polygon, regular polygon sides=4,draw=black,fill=red,scale=0.5]{};
\node at (-0.25,-0.25)[regular polygon, regular polygon sides=4,draw=black,fill=red,scale=0.5]{};
\node at (-0.25,0.25)[regular polygon, regular polygon sides=4,draw=black,fill=red,scale=0.5]{};
\node at (0.25,-0.25)[regular polygon, regular polygon sides=4,draw=black,fill=red,scale=0.5]{};

\draw (-0.75,0.25+6) -- (0.75,0.25+6);
\draw (-0.75,0.75+6) -- (0.75,0.75+6);
\draw (-0.75,-0.25+6) -- (0.75,-0.25+6);
\draw (-0.75,-0.75+6) -- (0.75,-0.75+6);
\draw (0.25,-0.75+6) -- (0.25,0.75+6);
\draw (0.75,-0.75+6) -- (0.75,0.75+6);
\draw (-0.25,-0.75+6) -- (-0.25,0.75+6);
\draw (-0.75,-0.75+6) -- (-0.75,0.75+6);

\node at (-0.25,0.25+6)[regular polygon, regular polygon sides=4,draw=black,fill=red,scale=0.5]{};
\node at (0.25,-0.25+6)[regular polygon, regular polygon sides=4,draw=black,fill=red,scale=0.5]{};
\node at (-0.25,-0.25+6)[regular polygon, regular polygon sides=4,draw=black,fill=red,scale=0.5]{};
\node at (0.25,0.25+6)[regular polygon, regular polygon sides=4,draw=black,fill=red,scale=0.5]{};

\end{tikzpicture}
\end{center}
\caption{A set $S$ of particles, which is electable (on the left) and a set $S'$, which is not electable (on the right; circle: particle in $S$; square: particle in $S'$).}
\label{twosets}
\end{figure}
\begin{prop}\label{existamere}
In any isometric subgraph $G$ of $\mathcal{F}_i$ such that $|V(G)|\ge 2$, there exist at least two vertices $u_{+}$ and $u_{-}$ satisfying conditions I) and II) of Definition \ref{defcontractible}, for $S=V(G)$.
\end{prop}
\begin{proof}
If there are two vertices of degree 1, they easily satisfy conditions I) and II) of Definition \ref{defcontractible}. Thus, we can suppose that there is at most one vertex of degree 1.
Let $u_+$ ($u_{-}$, respectively) be the vertex at position $(x_m,y_+,i)$ (at position $(x_m,y_{-},i)$, respectively) such that $(x_m,y_+,i)$ maximizes $y$ ( $(x_m,y_-,i)$ minimizes $y$, respectively) in the set $X=\{(x,y,i)\in S|\ \forall (x',y',i)\in S,\  x\ge x' \}$, i.e., the set of particles having a maximum (in the layer $i$) first coordinate. In the case  $|X|=1$, we take $u_+$ and $u_-$ the same way but in the set $Y=\{(x,y,i)\in S|\ \forall (x',y',i)\in S,\  x\le x' \}$, i.e., the set of particles having a minimum (in the layer $i$) first coordinate. Remark that $|X|=|Y|=1$ contradicts the fact that  there is at most one vertex of degree 1. Hence at least one of $X$ and $Y$ is not a singleton.
Moreover, since $G$ is isometric, both $X$ and $Y$ induce paths in $G$ and thus $u_+$ and $u_{-}$ have degree $2$. This implies that both $u_+$ and $u_{-}$ satisfy condition II).

For condition I), assume first that $|X|\ge 2$. Suppose, by contradiction, that one vertex among $u_+$ and $u_{-}$ does not satisfy condition I) and suppose, without loss of generality, that this vertex is $u_+$.
It implies that $(x_m-1,y_{+}-1,i)\notin S$ and $(x_m-1,y_{+},i)\in S$. Since $u_-$ has degree $2$, then $(x_m-1,y_{-},i)\in S$. But this contradicts the fact that $G$ is isometric since then we should have a path in $S$ between $(x_m-1,y_{+},i)$ and $(x_m-1,y_{-},i)$ going through $(x_m-1,y_{+}-1,i)$.
Second, if $|X|=1$, by symmetry, the same contradiction is obtained for $Y$. Consequently, both $u_+$ and $u_{-}$ satisfy condition I).

\end{proof}
\subsubsection{Electable sets and correctness for electable sets}

\begin{de}
For a set $S\subseteq V(\mathcal{F})$, the graph $G_S$ is the graph with vertex set the subsets of vertices $A$ such that $A$ induces a connected component (maximal connected subgraph) of $\mathcal{F}[S\cap V(\mathcal{F}_i)]$, for an integer $i$. There is an edge between two vertices $A,B$ of $G_S$ if there exists an edge in $\mathcal{F}$ between a vertex of $A$ and a vertex of $B$.
\end{de}

We say that $S$ is {\em electable} if the following three conditions are satisfied by $S$:
\begin{enumerate}
\item[a)] $G_S$ is a tree;
\item[b)] for every $A\in V(G_S)$, $\mathcal{F}[A]$ is isometric;
\item[c)] for every two adjacent sets $A$ and $B$ in $G_S$, the sets $\{u\in B | \ \exists v \in A,\ v\in N(u)\}$ and $\{u\in A | \ \exists v \in B,\ v\in N(u)\}$ both induce connected subraphs in $\mathcal{F}$.
\end{enumerate}

Figure \ref{twosets} illustrates a set $S$ of particles which is electable and a set $S'$ which is not electable. Note that condition a) is not satisfied by $S'$ ($G_{S'}$ is not a tree), but $S'$ satisfies both conditions b) and c).

\begin{algorithm}
\caption{The leader election algorithm in the heterogeneous case for a particle $p$, with $S_C$ being the set of particles in state \textbf{C}.} 
\label{alg2}
\begin{algorithmic} 
\State \textbf{Case 1: } State \textbf{C}
\If {$p$ is $S_C$-contractible}
    \If {$p$ has no neighbor in $S_C$}
        \State set the state to \textbf{L}
    \Else
        \State set the state to \textbf{N}
    \EndIf
\Else
        \State  stay in state \textbf{C}
\EndIf
\State \textbf{Case 2: } States \textbf{L} or \textbf{N} \\ Perform no further actions
\end{algorithmic}
\end{algorithm}

\begin{prop}\label{prop1}
If $S$ is electable then there is always an $S$-contractible particle in $S$.
\end{prop}
\begin{proof}
Let $A$ be a leaf of $G_S$ and let $B$ the neighbor of $A$ in $G_S$. Since $A$ is a leaf, either no particle in $A$ is adjacent with particles from the layer below $A$ or no particle in $A$ is adjacent with particles from the layer above $A$. We suppose, without loss of generality, that every particles which is adjacent to a particle of $A$ is either in $A$ or in the layer above $A$.
If $A$ contains only one vertex, then the only particle of $A$ is $S$-contractible, since conditions I), II) and III) are satisfied for this particle. Otherwise, by Proposition \ref{existamere}, there exist two particles $p_+$ and $p_-$ satisfying both conditions I) and II).

If $|N_S^{\pm}(p_{+})|\ge 3$, then $\mathcal{F}[N_S^{\pm}(p_{+})]$ is connected and, since $|A|\ge 2$, there exists a neighbor $p'$ of $p_+$ in $A$ such that $N_S^{\pm}(p_+) \cap N_S^{\pm}(p') \neq\emptyset$ (note that for any neighbor $q$ of $p^+$ in $A$ and for any set of three particles of $N_S^{\pm}(p_{+})\cap B$, there exists a particle $q'$ in this set adjacent to $q$). Consequently, condition III) is satisfied by $p$.

If $|N_S^{\pm}(p_+)|= 2$, then, $\mathcal{F}[N_S^{\pm}(p_+)]$ is connected since, otherwise, $\mathcal{F}[B]$ would not be isometric. The same fact holds for $p_{-}$: if $|N_S^{\pm}(p_-)|= 2$, then  $\mathcal{F}[N_S^{\pm}(p_-)]$ is connected.

Moreover, if $|N_S^{\pm}(p_+)|= 0$ or $|N_S^{\pm}(p_-)|= 0$, then, $p_+$ or $p_-$ satisfies, by definition, condition III).

Finally, we are left with the case that both $0<|N_S^{\pm}(p_+)|\le 2$ and $0<|N_S^{\pm}(p_-)|\le 2$ and both $\mathcal{F}[N_S^{\pm}(p_+)]$ and $\mathcal{F}[N_S^{\pm}(p_-)]$ are connected. Let $q_+$ be a vertex in $N_S^{\pm}(p_+)$ and let $q_-$ be a vertex in $N_S^{\pm}(p_-)$.
Since $\mathcal{F}[B]$ is isometric, there should be an isometric path between $q_-$ and $q_+$ in $B$. Also, since $S$ satisfies Property c), this path should be in the set $\{u\in B | \ \exists v \in A,\ v\in N(u)\}$. Consequently, either there exists a neighbor $p'_+\in A$ of $p_+$ such that $N_S^{\pm}(p_+) \cap N_S^{\pm}(p'_+) =\emptyset$ or there exist a neighbor $p'_-\in 	A$ of $p_-$ such that $N_S^{\pm}(p_-) \cap N_S^{\pm}(p'_-) =\emptyset$, hence condition III) is satisfied. 
\end{proof}

\begin{prop}\label{prop2}
If $S$ is electable, then $S\setminus \{p\}$ is also electable, for $p$ an $S$-contractible particle.
\end{prop}
\begin{proof}
Suppose that $p$ is in a vertex $A$ of $G_S$. 
First, if there exits a vertex $p'\in N_S^0(p)$ such that $N_S^{\pm}(p) \cap N_S^{\pm}(p') \neq\emptyset$, then $S\setminus \{p\}$ satisfies Property a) ($G_{S\setminus\{p\}}$ remains a tree). Otherwise, we have either $N_s^0(p)=\emptyset$ or $N_s^0(p)\neq \emptyset$ and $N_S^{\pm}(p)=\emptyset$. In both cases, it can be noted that $G_{S\setminus\{p\}}$ remains a tree and, consequently, that $S\setminus \{p\}$ satisfies Property a).

Second, since $\mathcal{F}[M_S(p)]$ is connected and $|N_S^0(p)|\le 2$, every path passing by $p$ will be as short as previously in $A\setminus \{p\}$. Consequently, $A$ remains isometric and $S\setminus \{p\}$ satisfies Property b).

Third, if there exits a vertex $p'\in N_S^0(p)$ such that $N_S^{\pm}(p) \cap N_S^{\pm}(p') \neq\emptyset$, then $S\setminus \{p\}$ satisfies Property c) since both $\mathcal{F}[N_S^{\pm}(p)]$ and $\mathcal{F}[M_S(p)]$ are connected. Now suppose that there does not exit a vertex $p'\in N_S^0(p)$ such that $N_S^{\pm}(p) \cap N_S^{\pm}(p') \neq\emptyset$. If $N_s^0(p)=\emptyset$, then since $\mathcal{F}[N_S^{\pm}(p)]$ is connected, $S\setminus \{p\}$ satisfies Property c). Finally, If $N_s^{\pm}(p)=\emptyset$, then, since $\mathcal{F}[M_S(p)]$ is connected, $S\setminus \{p\}$ satisfies Property c)
\end{proof}

We finish this part by giving the following result about correctness of our algorithm.

\begin{theo}
Let $S$ be the set of vertices occupied by particles.
If $S$ is electable, then at the end of the execution of Algorithm \ref{alg2}, there is exactly one particle in state \textbf{L}.
\end{theo}
\begin{proof}
By Proposition \ref{prop1}, there is always one $S$-contractible particle $p$ in $S$.
By Proposition \ref{prop2}, $S\setminus \{p\}$ is also electable. Consequently, by combining theses two results we obtain that the theorem holds.
\end{proof}

\begin{figure}[t]
\begin{center}
\begin{tikzpicture}[scale=0.80]

\draw (-1,1+3) -- (1,1+3);
\draw (-1,0+3) -- (1,0+3);
\draw (-1,0.5+3) -- (1,0.5+3);
\draw (-1,-0.5+3) -- (1,-0.5+3);
\draw (-1,-1+3) -- (1,-1+3);
\draw (0,-1+3) -- (0,1+3);
\draw (0.5,-1+3) -- (0.5,1+3);
\draw (-0.5,-1+3) -- (-0.5,1+3);
\draw (1,-1+3) -- (1,1+3);
\draw (-1,-1+3) -- (-1,1+3);

\node at (-1.5,3){$\mathcal{F}_0$};
\node at (-1.5,0){$\mathcal{F}_{-1}$};
\node at (-1.5,6){$\mathcal{F}_{1}$};
\node at (0,-1.1){\scriptsize{$r=0$}};
\node at (3,-1.1){\scriptsize{$r=1$}};
\node at (6,-1.1){\scriptsize{$r=2$}};
\node at (9,-1.1){\scriptsize{$r=3$}};
\node at (12,-1.1){\scriptsize{$r=4$}};
\node at (15,-1.1){\scriptsize{$r=5$}};

\node at (0.5,0+2.5)[circle,draw=black,fill=black,scale=0.5]{};
\draw (-0.75,0.25) -- (0.75,0.25);
\draw (-0.75,0.75) -- (0.75,0.75);
\draw (-0.75,-0.25) -- (0.75,-0.25);
\draw (-0.75,-0.75) -- (0.75,-0.75);
\draw (0.25,-0.75) -- (0.25,0.75);
\draw (0.75,-0.75) -- (0.75,0.75);
\draw (-0.25,-0.75) -- (-0.25,0.75);
\draw (-0.75,-0.75) -- (-0.75,0.75);

\node at (0.25,0.25)[regular polygon, regular polygon sides=4,draw=black,fill=red,scale=0.5]{};
\node at (-0.25,-0.25)[regular polygon, regular polygon sides=4,draw=black,fill=red,scale=0.5]{};
\node at (-0.25,0.25)[regular polygon, regular polygon sides=4,draw=black,fill=red,scale=0.5]{};
\node at (0.25,-0.25)[circle,draw=black,fill=black,scale=0.5]{};

\draw (-0.75,0.25+6) -- (0.75,0.25+6);
\draw (-0.75,0.75+6) -- (0.75,0.75+6);
\draw (-0.75,-0.25+6) -- (0.75,-0.25+6);
\draw (-0.75,-0.75+6) -- (0.75,-0.75+6);
\draw (0.25,-0.75+6) -- (0.25,0.75+6);
\draw (0.75,-0.75+6) -- (0.75,0.75+6);
\draw (-0.25,-0.75+6) -- (-0.25,0.75+6);
\draw (-0.75,-0.75+6) -- (-0.75,0.75+6);

\node at (-0.25,0.25+6)[regular polygon, regular polygon sides=4,draw=black,fill=red,scale=0.5]{};
\node at (0.25,-0.25+6)[circle,draw=black,fill=black,scale=0.5]{};
\node at (-0.25,-0.25+6)[regular polygon, regular polygon sides=4,draw=black,fill=red,scale=0.5]{};
\node at (0.25,0.25+6)[regular polygon, regular polygon sides=4,draw=black,fill=red,scale=0.5]{};

\draw (-1+3,1+3) -- (1+3,1+3);
\draw (-1+3,0+3) -- (1+3,0+3);
\draw (-1+3,0.5+3) -- (1+3,0.5+3);
\draw (-1+3,-0.5+3) -- (1+3,-0.5+3);
\draw (-1+3,-1+3) -- (1+3,-1+3);
\draw (0+3,-1+3) -- (0+3,1+3);
\draw (0.5+3,-1+3) -- (0.5+3,1+3);
\draw (-0.5+3,-1+3) -- (-0.5+3,1+3);
\draw (1+3,-1+3) -- (1+3,1+3);
\draw (-1+3,-1+3) -- (-1+3,1+3);

\node at (0.5+3,0+2.5)[circle,draw=black,fill=black,scale=0.5]{};
\draw (-0.75+3,0.25) -- (0.75+3,0.25);
\draw (-0.75+3,0.75) -- (0.75+3,0.75);
\draw (-0.75+3,-0.25) -- (0.75+3,-0.25);
\draw (-0.75+3,-0.75) -- (0.75+3,-0.75);
\draw (0.25+3,-0.75) -- (0.25+3,0.75);
\draw (0.75+3,-0.75) -- (0.75+3,0.75);
\draw (-0.25+3,-0.75) -- (-0.25+3,0.75);
\draw (-0.75+3,-0.75) -- (-0.75+3,0.75);

\node at (0.25+3,0.25)[circle,draw=black,fill=black,scale=0.5]{};
\node at (-0.25+3,-0.25)[regular polygon, regular polygon sides=3,draw=black,fill=green,scale=0.35]{};
\node at (-0.25+3,0.25)[regular polygon, regular polygon sides=4,draw=black,fill=red,scale=0.5]{};
\node at (0.25+3,-0.25)[circle,draw=black,fill=black,scale=0.5]{};

\draw (-0.75+3,0.25+6) -- (0.75+3,0.25+6);
\draw (-0.75+3,0.75+6) -- (0.75+3,0.75+6);
\draw (-0.75+3,-0.25+6) -- (0.75+3,-0.25+6);
\draw (-0.75+3,-0.75+6) -- (0.75+3,-0.75+6);
\draw (0.25+3,-0.75+6) -- (0.25+3,0.75+6);
\draw (0.75+3,-0.75+6) -- (0.75+3,0.75+6);
\draw (-0.25+3,-0.75+6) -- (-0.25+3,0.75+6);
\draw (-0.75+3,-0.75+6) -- (-0.75+3,0.75+6);

\node at (-0.25+3,0.25+6)[regular polygon, regular polygon sides=4,draw=black,fill=red,scale=0.5]{};
\node at (0.25+3,-0.25+6)[circle,draw=black,fill=black,scale=0.5]{};
\node at (-0.25+3,-0.25+6)[circle,draw=black,fill=black,scale=0.5]{};
\node at (0.25+3,0.25+6)[regular polygon, regular polygon sides=3,draw=black,fill=green,scale=0.35]{};

\draw (-1+6,1+3) -- (1+6,1+3);
\draw (-1+6,0+3) -- (1+6,0+3);
\draw (-1+6,0.5+3) -- (1+6,0.5+3);
\draw (-1+6,-0.5+3) -- (1+6,-0.5+3);
\draw (-1+6,-1+3) -- (1+6,-1+3);
\draw (0+6,-1+3) -- (0+6,1+3);
\draw (0.5+6,-1+3) -- (0.5+6,1+3);
\draw (-0.5+6,-1+3) -- (-0.5+6,1+3);
\draw (1+6,-1+3) -- (1+6,1+3);
\draw (-1+6,-1+3) -- (-1+6,1+3);

\node at (0.5+6,0+2.5)[circle,draw=black,fill=black,scale=0.5]{};
\draw (-0.75+6,0.25) -- (0.75+6,0.25);
\draw (-0.75+6,0.75) -- (0.75+6,0.75);
\draw (-0.75+6,-0.25) -- (0.75+6,-0.25);
\draw (-0.75+6,-0.75) -- (0.75+6,-0.75);
\draw (0.25+6,-0.75) -- (0.25+6,0.75);
\draw (0.75+6,-0.75) -- (0.75+6,0.75);
\draw (-0.25+6,-0.75) -- (-0.25+6,0.75);
\draw (-0.75+6,-0.75) -- (-0.75+6,0.75);

\node at (0.25+6,0.25)[regular polygon, regular polygon sides=4,draw=black,fill=red,scale=0.5]{};
\node at (-0.25+6,-0.25)[regular polygon, regular polygon sides=3,draw=black,fill=green,scale=0.35]{};
\node at (-0.25+6,0.25)[regular polygon, regular polygon sides=3,draw=black,fill=green,scale=0.35]{};
\node at (0.25+6,-0.25)[circle,draw=black,fill=black,scale=0.5]{};

\draw (-0.75+6,0.25+6) -- (0.75+6,0.25+6);
\draw (-0.75+6,0.75+6) -- (0.75+6,0.75+6);
\draw (-0.75+6,-0.25+6) -- (0.75+6,-0.25+6);
\draw (-0.75+6,-0.75+6) -- (0.75+6,-0.75+6);
\draw (0.25+6,-0.75+6) -- (0.25+6,0.75+6);
\draw (0.75+6,-0.75+6) -- (0.75+6,0.75+6);
\draw (-0.25+6,-0.75+6) -- (-0.25+6,0.75+6);
\draw (-0.75+6,-0.75+6) -- (-0.75+6,0.75+6);

\node at (-0.25+6,0.25+6)[regular polygon, regular polygon sides=3,draw=black,fill=green,scale=0.35]{};
\node at (0.25+6,-0.25+6)[circle,draw=black,fill=black,scale=0.5]{};
\node at (-0.25+6,-0.25+6)[regular polygon, regular polygon sides=4,draw=black,fill=red,scale=0.5]{};
\node at (0.25+6,0.25+6)[regular polygon, regular polygon sides=3,draw=black,fill=green,scale=0.35]{};

\draw (-1+9,1+3) -- (1+9,1+3);
\draw (-1+9,0+3) -- (1+9,0+3);
\draw (-1+9,0.5+3) -- (1+9,0.5+3);
\draw (-1+9,-0.5+3) -- (1+9,-0.5+3);
\draw (-1+9,-1+3) -- (1+9,-1+3);
\draw (0+9,-1+3) -- (0+9,1+3);
\draw (0.5+9,-1+3) -- (0.5+9,1+3);
\draw (-0.5+9,-1+3) -- (-0.5+9,1+3);
\draw (1+9,-1+3) -- (1+9,1+3);
\draw (-1+9,-1+3) -- (-1+9,1+3);

\node at (0.5+9,0+2.5)[circle,draw=black,fill=black,scale=0.5]{};
\draw (-0.75+9,0.25) -- (0.75+9,0.25);
\draw (-0.75+9,0.75) -- (0.75+9,0.75);
\draw (-0.75+9,-0.25) -- (0.75+9,-0.25);
\draw (-0.75+9,-0.75) -- (0.75+9,-0.75);
\draw (0.25+9,-0.75) -- (0.25+9,0.75);
\draw (0.75+9,-0.75) -- (0.75+9,0.75);
\draw (-0.25+9,-0.75) -- (-0.25+9,0.75);
\draw (-0.75+9,-0.75) -- (-0.75+9,0.75);

\node at (0.25+9,0.25)[regular polygon, regular polygon sides=4,draw=black,fill=red,scale=0.5]{};
\node at (-0.25+9,-0.25)[regular polygon, regular polygon sides=3,draw=black,fill=green,scale=0.35]{};
\node at (-0.25+9,0.25)[regular polygon, regular polygon sides=3,draw=black,fill=green,scale=0.35]{};
\node at (0.25+9,-0.25)[circle,draw=black,fill=black,scale=0.5]{};

\draw (-0.75+9,0.25+6) -- (0.75+9,0.25+6);
\draw (-0.75+9,0.75+6) -- (0.75+9,0.75+6);
\draw (-0.75+9,-0.25+6) -- (0.75+9,-0.25+6);
\draw (-0.75+9,-0.75+6) -- (0.75+9,-0.75+6);
\draw (0.25+9,-0.75+6) -- (0.25+9,0.75+6);
\draw (0.75+9,-0.75+6) -- (0.75+9,0.75+6);
\draw (-0.25+9,-0.75+6) -- (-0.25+9,0.75+6);
\draw (-0.75+9,-0.75+6) -- (-0.75+9,0.75+6);

\node at (-0.25+9,0.25+6)[regular polygon, regular polygon sides=3,draw=black,fill=green,scale=0.35]{};
\node at (0.25+9,-0.25+6)[circle,draw=black,fill=black,scale=0.5]{};
\node at (-0.25+9,-0.25+6)[regular polygon, regular polygon sides=4,draw=black,fill=red,scale=0.5]{};
\node at (0.25+9,0.25+6)[regular polygon, regular polygon sides=3,draw=black,fill=green,scale=0.35]{};

\draw (-1+12,1+3) -- (1+12,1+3);
\draw (-1+12,0+3) -- (1+12,0+3);
\draw (-1+12,0.5+3) -- (1+12,0.5+3);
\draw (-1+12,-0.5+3) -- (1+12,-0.5+3);
\draw (-1+12,-1+3) -- (1+12,-1+3);
\draw (0+12,-1+3) -- (0+12,1+3);
\draw (0.5+12,-1+3) -- (0.5+12,1+3);
\draw (-0.5+12,-1+3) -- (-0.5+12,1+3);
\draw (1+12,-1+3) -- (1+12,1+3);
\draw (-1+12,-1+3) -- (-1+12,1+3);

\node at (0.5+12,0+2.5)[circle,draw=black,fill=black,scale=0.5]{};
\draw (-0.75+12,0.25) -- (0.75+12,0.25);
\draw (-0.75+12,0.75) -- (0.75+12,0.75);
\draw (-0.75+12,-0.25) -- (0.75+12,-0.25);
\draw (-0.75+12,-0.75) -- (0.75+12,-0.75);
\draw (0.25+12,-0.75) -- (0.25+12,0.75);
\draw (0.75+12,-0.75) -- (0.75+12,0.75);
\draw (-0.25+12,-0.75) -- (-0.25+12,0.75);
\draw (-0.75+12,-0.75) -- (-0.75+12,0.75);

\node at (0.25+12,0.25)[regular polygon, regular polygon sides=3,draw=black,fill=green,scale=0.35]{};
\node at (-0.25+12,-0.25)[regular polygon, regular polygon sides=3,draw=black,fill=green,scale=0.35]{};
\node at (-0.25+12,0.25)[regular polygon, regular polygon sides=3,draw=black,fill=green,scale=0.35]{};
\node at (0.25+12,-0.25)[regular polygon, regular polygon sides=4,draw=black,fill=red,scale=0.5]{};

\draw (-0.75+12,0.25+6) -- (0.75+12,0.25+6);
\draw (-0.75+12,0.75+6) -- (0.75+12,0.75+6);
\draw (-0.75+12,-0.25+6) -- (0.75+12,-0.25+6);
\draw (-0.75+12,-0.75+6) -- (0.75+12,-0.75+6);
\draw (0.25+12,-0.75+6) -- (0.25+12,0.75+6);
\draw (0.75+12,-0.75+6) -- (0.75+12,0.75+6);
\draw (-0.25+12,-0.75+6) -- (-0.25+12,0.75+6);
\draw (-0.75+12,-0.75+6) -- (-0.75+12,0.75+6);

\node at (-0.25+12,0.25+6)[regular polygon, regular polygon sides=3,draw=black,fill=green,scale=0.35]{};
\node at (0.25+12,-0.25+6)[regular polygon, regular polygon sides=4,draw=black,fill=red,scale=0.5]{};
\node at (-0.25+12,-0.25+6)[regular polygon, regular polygon sides=3,draw=black,fill=green,scale=0.35]{};
\node at (0.25+12,0.25+6)[regular polygon, regular polygon sides=3,draw=black,fill=green,scale=0.35]{};

\draw (-1+15,1+3) -- (1+15,1+3);
\draw (-1+15,0+3) -- (1+15,0+3);
\draw (-1+15,0.5+3) -- (1+15,0.5+3);
\draw (-1+15,-0.5+3) -- (1+15,-0.5+3);
\draw (-1+15,-1+3) -- (1+15,-1+3);
\draw (0+15,-1+3) -- (0+15,1+3);
\draw (0.5+15,-1+3) -- (0.5+15,1+3);
\draw (-0.5+15,-1+3) -- (-0.5+15,1+3);
\draw (1+15,-1+3) -- (1+15,1+3);
\draw (-1+15,-1+3) -- (-1+15,1+3);

\node at (0.5+15,0+2.5)[regular polygon, regular polygon sides=5,draw=black,fill=blue,scale=0.7]{};
\draw (-0.75+15,0.25) -- (0.75+15,0.25);
\draw (-0.75+15,0.75) -- (0.75+15,0.75);
\draw (-0.75+15,-0.25) -- (0.75+15,-0.25);
\draw (-0.75+15,-0.75) -- (0.75+15,-0.75);
\draw (0.25+15,-0.75) -- (0.25+15,0.75);
\draw (0.75+15,-0.75) -- (0.75+15,0.75);
\draw (-0.25+15,-0.75) -- (-0.25+15,0.75);
\draw (-0.75+15,-0.75) -- (-0.75+15,0.75);

\node at (0.25+15,0.25)[regular polygon, regular polygon sides=3,draw=black,fill=green,scale=0.35]{};
\node at (-0.25+15,-0.25)[regular polygon, regular polygon sides=3,draw=black,fill=green,scale=0.35]{};
\node at (-0.25+15,0.25)[regular polygon, regular polygon sides=3,draw=black,fill=green,scale=0.35]{};
\node at (0.25+15,-0.25)[regular polygon, regular polygon sides=3,draw=black,fill=green,scale=0.35]{};

\draw (-0.75+15,0.25+6) -- (0.75+15,0.25+6);
\draw (-0.75+15,0.75+6) -- (0.75+15,0.75+6);
\draw (-0.75+15,-0.25+6) -- (0.75+15,-0.25+6);
\draw (-0.75+15,-0.75+6) -- (0.75+15,-0.75+6);
\draw (0.25+15,-0.75+6) -- (0.25+15,0.75+6);
\draw (0.75+15,-0.75+6) -- (0.75+15,0.75+6);
\draw (-0.25+15,-0.75+6) -- (-0.25+15,0.75+6);
\draw (-0.75+15,-0.75+6) -- (-0.75+15,0.75+6);

\node at (-0.25+15,0.25+6)[regular polygon, regular polygon sides=3,draw=black,fill=green,scale=0.35]{};
\node at (0.25+15,-0.25+6)[regular polygon, regular polygon sides=3,draw=black,fill=green,scale=0.35]{};
\node at (-0.25+15,-0.25+6)[regular polygon, regular polygon sides=3,draw=black,fill=green,scale=0.35]{};
\node at (0.25+15,0.25+6)[regular polygon, regular polygon sides=3,draw=black,fill=green,scale=0.35]{};
\end{tikzpicture}
\end{center}
\caption{An example of the execution of Algorithm~\ref{alg2} with successive rounds ($r$ being the number of finished rounds) going from left to right on an electable set $S$ (circle and square: particle in $S$; square: $S_C$-contractible particle of $S$; triangle: particle which is no more in $S_C$; pentagon: elected leader).}
\label{examplexec}
\end{figure}
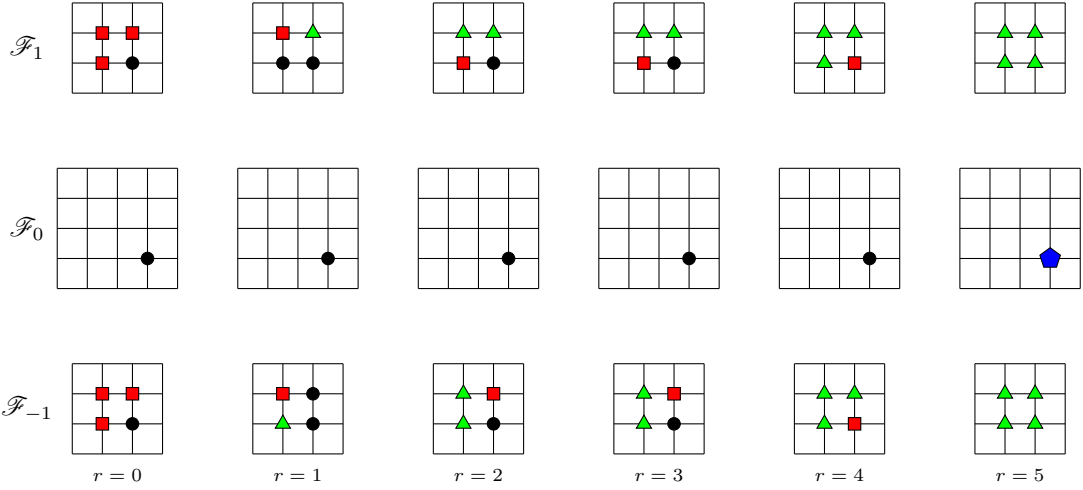

\subsubsection{Details about Algorithm \ref{alg2}}
We begin with giving some intuition on how the algorithm works. Each particle $p$ verifies properties in its neighborhood which guarantee that $p$ is not an articulation (the removing of it does not split the network into two connected components) and $p$ is no more candidate to be leader if the properties are satisfied. In this algorithm, only the candidate particles are considered. This implies that a particle $p$ can become no longer an articulation in the case some particles in its neighborhood are no longer candidate.
The leader is the last remaining particle. For example, if there are only three particles which are at positions $(-1,0,0)$, $(0,0,0)$ and $(1,0,0)$, then there is only one articulation which the particle at position $(0,0,0)$. Algorithm~\ref{alg2} will then first remove the candidacy of one of the two particles at positions $(-1,0,0)$ and $(1,0,0)$ and after will remove one of the two remaining particles and will elect the last one (which will be candidate to be leader).

More precisely, let $S$ be the set of particles in state \textbf{C}. Algorithm~\ref{alg2} consists in removing from $S$ some particles of the layer $i$ which are both on the geographical border of $G[S\cap \mathcal{F}_i]$ and not articulations of $\mathcal{F}[S]$ (an articulation being a vertex which split, when it is removed, in two connected components the vertices of $\mathcal{F}[S]$). 
In this algorithm we never remove a particle of the layer $i$ which is an articulation of $G[S\cap \mathcal{F}_i]$ and we never remove a particle having neighbors in both the layer $i-1$ and $i+1$ (such particles could be articulations). The algorithm will finish when there is only one particle in state \textbf{C}. The fact that $S$ is electable implies that there is always a contractible particle , i.e., a particle that we can remove (see Proposition \ref{prop1}).

Figure \ref{examplexec} illustrates an example of the execution of Algorithm~\ref{alg2} on the first five rounds on a set $S$. In this figure, $S$ is electable since it can be easily noticed that $G_S$ is a path of three vertices and that every $A\in V(G_S)$ is such that $\mathcal{F}[A]$ is isometric. Also, it can be remarked that another example of the execution of Algorithm~\ref{alg2} could lead to the exclusion from $S_C$ of the particle in the top-left corner of the layer $-1$ or $1$.
\subsubsection{Required number of rounds, space-complexity and algorithm to check if a set is electable}

In Algorithm \ref{alg2}, it can be noted that each particle $p$ has only to compute if it is $S$-contractible or not and, as explained before Proposition \ref{existamere}, such verifications can be done by only using the labels of the ports of $p$ in contact with particles of the neighborhood and, for each pair of particles $q$ and $q'$ in the closed neighborhood of $p$, the labels of the ports of $q$ and $q'$ in contact with a common neighbor. Thus, in order that Algorithm \ref{alg2} works properly, it only requires that each particle is able to compute $O(1)$ operations. Moreover, about the required memory space and required number of rounds for execution, we state the following. 
\begin{prop}
Whatever the structure of $\mathcal{F}[S]$, if $\mathcal{F}[S]$ is connected and $S$ is electable, then one particle is in state \textbf{L} after $|S|$ rounds since the beginning of the execution of Algorithm \ref{alg2}. Moreover, the memory space used by Algorithm \ref{alg2} is constant.
\end{prop}
\begin{proof}
First, it can be noted, that by Propositions~\ref{prop1} and ~\ref{prop2}, there is always (at least) one $S$-contractible particle $p$ in $S$. Thus, after each round there is one less particle in state \textbf{C}. Consequently, after $|S|$ rounds, the last particle that was in state \textbf{C} will change its state to \textbf{L}.
Second, only the information to check if a particle is $S$-contractible are required in the memory. These information consists in the labels of the ports of the particle in contact with particles of the neighborhood and, for each pair of particles $q$ and $q'$ in the closed neighborhood, in the labels of the ports of $q$ and $q'$ in contact with a common neighbor. Since the number of neighbors is bounded, the required memory space is also bounded.
\end{proof}


The vertices of $S$ on a layer $k$ are said to {\em form a circle} if there exist a vertex $u\in S$ and a positive integer $d$ such that $v\in S$ if and only if $u$ is at distance at most $d$ from $v$ in $\mathcal{F}$.
Analogously, the vertices on a layer $k$ of $S$ {\em form a rectangle} if there exist four integers $i_0,$ $i_1$, $j_0$ and $j_1$ such that $(i,j,k)\in S$ if and only if $i_0\le i\le i_1$ and $j_0\le j\le j_1$.

It can be easily noted that in the case the vertices of $S$ on a layer $k$ forms either a circle or a rectangle, the graph induced by the vertices of $S$ on a layer $k$ is isometric (Property b)). If the vertices of $S$ form either a circle or a rectangle on every layer of $\mathcal{F}$, then $G_S$ is a path (Property a)) and since for every two adjacent sets $A$ and $B$ from $G_S$, the sets $\{u\in B | \ \exists v \in A,\ v\in N(u)\}$ and $\{u\in A | \ \exists v \in B,\ v\in N(u)\}$ both induce connected sets, it implies that $S$ is electable in this case. Thus, we obtain the following.

\begin{cor}
Let $S$ be the set of vertices of $\mathcal{F}$ occupied by particles.
If the vertices of $S$ form either a circle or a rectangle on every layer of $\mathcal{F}$, then Algorithm \ref{alg2} allows to elect a unique leader.
\end{cor}

Finally, we end this section by giving a non distributed algorithm in order to check if a set of particles $S$ is electable or not.

\begin{prop}
Given a set of particles $S$, there exists an $O(|S|^3)$ algorithm to verify if the set $S$ is electable or not.
\end{prop}
\begin{proof}
Let $L$ the set of the particles from $S$.
This algorithm can be decomposed into three steps, each step consisting in verifying one of the properties among properties a), b) and c).
The first step consists in separating the set $L$ into subsets depending on the value of $k$, i.e., the value of the third element of the triplets in $L$. Afterward, we split these subsets into connected components and we denote by $M$ this set of sets. We can easily build $G_S$ from $M$ (since the adjacency between two triplets is easy to compute). Finally, by checking if $G_S$ is a tree, we can determine if $S$ satisfy Property a) or not.

The second step consists in computing the distance in $\mathcal{F}[A]$ between every two triplets in a same set $A\in M$ and compare it with the distance given by the formula of Section~\ref{sec:2.1}. If there exists two triplets for which there is a difference between the two distances, then $S$ does not satisfy Property b), otherwise, it satisfies Property b).

The last step consists, for every pair of sets $(A,A')$ from $M$ such that there are two adjacent triplets $u\in A$ and $v\in A'$, in computing two sets $B$ and $B'$ and in checking if $B$ and $B'$ are connected. The set $B$ being the set of triplets of $A$ adjacent to a triplet of $A'$ and the set $B'$ being the set of triplets in $A'$ adjacent to a triplet of $A$. If at least one constructed set is not connected, then $S$ does not satisfy Property c), otherwise, it satisfies Property c).

Since computing the distance between every two triplets can be obtained by an $O(|S|^3)$ algorithm and since the other problems (number of connected component and cheking if a graph is a tree) can be solved by faster algorithms,  there exists an $O(|S|^3)$ algorithm to check if the set $S$ is electable or not.
\end{proof}

\section{Local and global identifiers}

In this section, we propose distributed and deterministic algorithms to create global and $\ell$-local identifiers of the particles. An $\ell$-local identifier is a variable affected to each particle of the network, which is different for every two particles that are at distance at most $\ell$.
Since particles have limited capacities, the idea of using $\ell$-local identifier is to allow to spare less memory than for global ones. 

The results presented in this part are in the continuity of the work about $\ell$-local identifiers in the context of programmable matter~\cite{GT2018} (this work was using colorings of the  $\ell$-th power of the  triangular grid).  Here, we will use the same method as well as a recent result on the combinatorial problem of coloring the  $\ell$-th power of the  face-centered cubic grid~\cite{GT2019}.

\subsection{Local identifier}

The main idea for computing local identifiers is to use a coloring of the $\ell$-th power of the face centered cubic grid as identifiers.
A \emph{$k$-coloring} of a graph $G$ is a map $f$ from $V(G)$ to $\{0,1,\ldots,k-1\}$ which satisfies $f(u)\neq f(v)$ for every $uv\in E(G)$.
The \emph{chromatic number} $\chi(G)$ of $G$, is the smallest integer $k$ such that there exists a $k$-coloring of $G$.
The \emph{$d$-th  power} $G^d$ of a graph $G$ is the graph obtained from $G$ by adding an edge between every two vertices satisfying $d_{G}(u,v)\le d$.

An upper bound on the chromatic number of the $d$-th power of $\mathcal{F}$ was given in~\cite{GT2019}:
\begin{theo}[\cite{GT2019}]
For any $d\ge 1$, there exists a coloring of the $d$-th power of $\mathcal{F}$ using $(d+1) \lceil (d+1)^{2}/2 \rceil$ colors.
\end{theo}

Moreover, the proof of the above theorem is constructive and the corresponding coloring is periodic, i.e., a pattern of fixed size is repeated in the whole grid to assign a color (integer) to each vertex of $\mathcal{F}$. This pattern can be described as follows.


Let $mod(\ell,k)$ be the integer $i$ such that $i\equiv \ell \pmod{k}$ and $0\le i \le k-1$. Let also $$m_{\ell}=\left\lceil \frac{(\ell +1)^2}2 \right\rceil,$$ and $$f_\ell (i,j,k)= mod(k, \ell+1) m_{\ell} + \left\{\begin{array}{ll}
mod(i+\ell j ,m_{\ell }) & \mbox{ if $k$ is even};\\
mod(i-0.5+\ell (j-0.5) ,m_{\ell }) & \mbox{ if $k$ is odd}.\end{array}\right.$$

We denote by $N_\ell(i,j,k)$ the triplet $(i',j',k')$, where $i'=mod(i, m_\ell)$, $j'=mod(j,m_\ell)$ and $k'=mod(k,\ell+1)$.

Then,  giving the color $f_\ell (N_{\ell}(i,j,k))$ to each vertex at position $(i,j,k)$, allows to obtain a coloring of the $\ell$-th power of the  face-centered cubic grid. 

Our global algorithm for computing the local identifiers works by, first, running a leader election algorithm to have a unique particle in state \textbf{L} and the other particles in state \textbf{N} (Algorithm~\ref{alg1} or~\ref{alg2}).
Second, when there is a leader, a spanning tree can be easily computed with a distributed algorithm (see~\cite{GT2018}).
Third,  in our proposed port renumbering algorithm (Algorithm~\ref{alg5}), we change the way the port are numbered in order that every particle has its ports numbered by the same number going in the same cardinal direction in $\mathcal{F}$ (it is needed in the heterogeneous case). 
Finally (see Algorithm \ref{alg4}), we give the identifier $(0,0,0)$ to the particle in state \textbf{L} and, for a particle $p$ receiving the message $(i,j,k)$ as first message, the inductive step consists in using messages to give the identifier $f_\ell (i,j,k)$ to $p$ and to send message $N_\ell (I(i,a),J(j,a),K(k,a))$ to each neighbor of $p$ connected through port $a$ of $p$. At the end, the variable $id$ of each particle contains its local identifier.

\begin{algorithm}
\caption{The port renumbering algorithm for a particle $p$.} 
\label{alg5}
\begin{algorithmic} 
\State \textbf{Case 1}: State \textbf{L}
\State for each port $a$ from $\text{child}(p)$ send a message $m_a$, containing $a$, through port $a$
\State \textbf{Case 2:} State \textbf{N}
\If {$p$ receives the message $m_b$, containing $b$, through port $a$}
    \State change the port number $a$ to $r(b)$ and changes the port numbers of the other ports following the same clockwise order
    \State update both $\text{parent}(p)$ and $\text{child}(p)$
    \State for each port $a$ from $\text{child}(p)$ send a message $m_a$, containing $a$, through port $a$
\EndIf
\end{algorithmic}
\end{algorithm}

Note that Algorithm~\ref{alg5} and Algorithm \ref{alg4} both send messages along the previously computed spanning tree. For this, it is assumed that for each particle $p$, we have two sets of ports $\text{parent}(p)$ and $\text{child}(p)$ which contains the port numbers of the particles in communication with its parent and with its children, respectively, in the spanning tree ($\text{parent}(p)$ is a singleton). 

The steps we use to compute the local identifiers are very similar with the ones of our previous work \cite{GT2018} on the 2-dimensional grid.

\begin{algorithm}
\caption{The $\ell$-local identifier algorithm for a particle $p$.} 
\label{alg4}
\begin{algorithmic} 
\State \textbf{Case 1: } State \textbf{L} \\ Set $i=0$, $j=0$, $k=0$ and $id=0$\\ 
For every port $a$ from $\text{child}(p)$, send the message $(I(i,a),J(j,a),K(k,a))$ to the neighbor of $p$ connected through port $a$

\State \textbf{Case 2: } State \textbf{N} 
\If {$p$ receives the message $(i',j',k')$ through port $a$}
        \State Set $i=i'$, $j=j'$, $k=k'$ and set $id= f_\ell (i,j,k)$
	\State For every port $a$ from $\text{child}(p)$, send the message $ N_\ell (I(i,a),J(j,a),K(k,a))$ to the neighbor of $p$ connected through port $a$
\EndIf
\end{algorithmic}
\end{algorithm}

The idea behind Algorithm~\ref{alg5} is to reproduce, in each particle, the way the ports are numbered in the leader particle. To achieve this goal, each particle $p$ receives a message from its parent containing the port number of the parent connected to $p$ and $p$ renumbers its own ports in order that its port numbers are coherent with the sent number. Figure~\ref{portreconfig}.a and Figure~\ref{portreconfig}.b illustrate the port numbers of particles before and after the execution of Algorithm~\ref{alg5}.
The function $r$ used in Algorithm~\ref{alg5} is defined as follows: $r(i)=(i+2)\pmod{4}$ if $i\in \{0,1,2,3\}$, $r(4)=10$, $r(5)=11$, $r(6)=8$, $r(7)=9$, $r(8)=6$, $r(9)=7$, $r(10)=4$ and $r(11)=5$.

Figure~\ref{portreconfig}.c illustrates the obtained 2-local identifiers after the execution of Algorithm~\ref{alg4}. In Figure~\ref{portreconfig}.a ,~\ref{portreconfig}.b and~\ref{portreconfig}.c , the edges of the spanning tree which are between vertices of different layers have been omitted in order to increase readability.

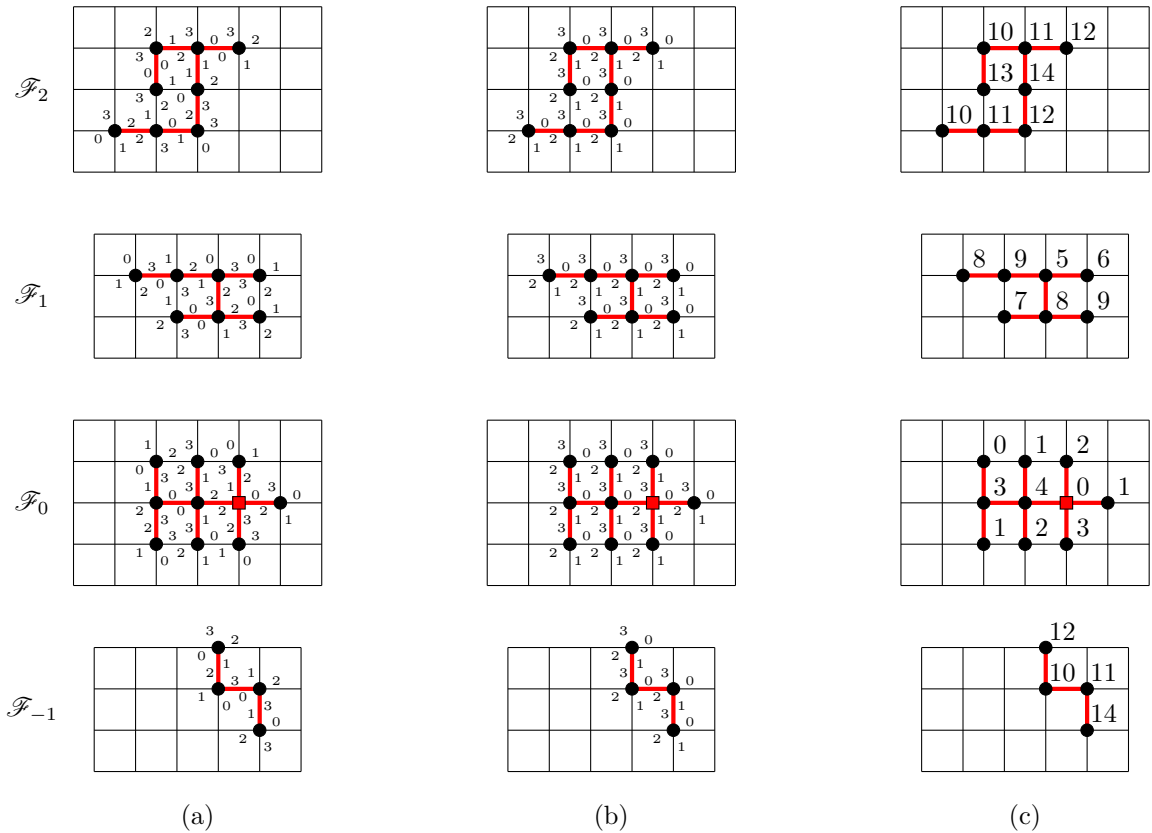
\begin{figure}[t]
\begin{center}
\begin{tikzpicture}[scale=1.1]

\draw (-2,1+2.5) -- (1,1+2.5);
\draw (-2,0+2.5) -- (1,0+2.5);
\draw (-2,0.5+2.5) -- (1,0.5+2.5);
\draw (-2,-0.5+2.5) -- (1,-0.5+2.5);
\draw (-2,-1+2.5) -- (1,-1+2.5);
\draw (0,-1+2.5) -- (0,1+2.5);
\draw (0.5,-1+2.5) -- (0.5,1+2.5);
\draw (-0.5,-1+2.5) -- (-0.5,1+2.5);
\draw (1,-1+2.5) -- (1,1+2.5);
\draw (-1,-1+2.5) -- (-1,1+2.5);
\draw (-1.5,-1+2.5) -- (-1.5,1+2.5);
\draw (-2,-1+2.5) -- (-2,1+2.5);

\draw[ultra thick, color=red] (-1,2.5) -- (0.5,2.5);
\draw[ultra thick, color=red] (-1,2) -- (-1,3);
\draw[ultra thick, color=red] (-0.5,2) -- (-0.5,3);
\draw[ultra thick, color=red] (0,2) -- (0,3);

\node at (-2.5,0){$\mathcal{F}_{-1}$};
\node at (-2.5,2.5){$\mathcal{F}_{0}$};
\node at (-2.5,5){$\mathcal{F}_{1}$};
\node at (-2.5,7.5){$\mathcal{F}_{2}$};

\node at (-0.5,-1.3){(a)};
\node at (4.5,-1.3){(b)};
\node at (9.5,-1.3){(c)};

\node at (0,2.5)[regular polygon, regular polygon sides=4,draw=black,fill=red,scale=0.5]{};
\node at (-0.5,2.5)[circle,draw=black,fill=black,scale=0.5]{};
\node at (0.5,2.5)[circle,draw=black,fill=black,scale=0.5]{};
\node at (0,2)[circle,draw=black,fill=black,scale=0.5]{};
\node at (0,3)[circle,draw=black,fill=black,scale=0.5]{};
\node at (-0.5,3)[circle,draw=black,fill=black,scale=0.5]{};
\node at (-1,3)[circle,draw=black,fill=black,scale=0.5]{};
\node at (-1,2.5)[circle,draw=black,fill=black,scale=0.5]{};
\node at (-1,2)[circle,draw=black,fill=black,scale=0.5]{};
\node at (-0.5,2)[circle,draw=black,fill=black,scale=0.5]{};

\node at (0.2-0.5,3.1-0.5) []{\tiny{2}};
\node at (0.7-0.5,2.6-0.5) []{\tiny{3}};
\node at (0.7-0.5,3.6-0.5) []{\tiny{1}};
\node at (0.7-0.5,3.1-0.5) []{\tiny{0}};
\node at (0.2+0.5,3.1-0.5) []{\tiny{0}};
\node at (0.7-1.5,2.6-0.5) []{\tiny{3}};
\node at (0.7-1.5,3.6-0.5) []{\tiny{2}};
\node at (0.7-1.5,3.1-0.5) []{\tiny{0}};
\node at (0.2-0.5,3.1) []{\tiny{0}};
\node at (0.2-0.5,3.1-1) []{\tiny{0}};

\node at (0.1-0.5,2.8-0.5) []{\tiny{1}};
\node at (0.6-0.5,2.3-0.5) []{\tiny{0}};
\node at (0.6-0.5,3.3-0.5) []{\tiny{2}};
\node at (0.6-0.5,2.8-0.5) []{\tiny{3}};
\node at (0.1+0.5,2.8-0.5) []{\tiny{1}};
\node at (0.6-1.5,2.3-0.5) []{\tiny{0}};
\node at (0.6-1.5,3.3-0.5) []{\tiny{3}};
\node at (0.6-1.5,2.8-0.5) []{\tiny{3}};
\node at (0.1-0.5,2.8-1) []{\tiny{1}};
\node at (0.1-0.5,2.8) []{\tiny{1}};

\node at (-0.2-0.5,2.9-0.5) []{\tiny{0}};
\node at (0.3-0.5,2.4-0.5) []{\tiny{1}};
\node at (0.3-0.5,3.4-0.5) []{\tiny{3}};
\node at (0.3-0.5,2.9-0.5) []{\tiny{2}};
\node at (-0.2+0.5,2.9-0.5) []{\tiny{2}};
\node at (0.3-1.5,2.4-0.5) []{\tiny{1}};
\node at (0.3-1.5,3.4-0.5) []{\tiny{0}};
\node at (0.3-1.5,2.9-0.5) []{\tiny{2}};
\node at (-0.2-0.5,2.9-0) []{\tiny{2}};
\node at (-0.2-0.5,2.9-1) []{\tiny{2}};

\node at (-0.1-0.5,3.2-0.5) []{\tiny{3}};
\node at (0.4-0.5,2.7-0.5) []{\tiny{2}};
\node at (0.4-0.5,3.7-0.5) []{\tiny{0}};
\node at (0.4-0.5,3.2-0.5) []{\tiny{1}};
\node at (-0.1+0.5,3.2-0.5) []{\tiny{3}};
\node at (0.4-1.5,2.7-0.5) []{\tiny{2}};
\node at (0.4-1.5,3.7-0.5) []{\tiny{1}};
\node at (0.4-1.5,3.2-0.5) []{\tiny{1}};
\node at (-0.1-0.5,3.2) []{\tiny{3}};
\node at (-0.1-0.5,3.2-1) []{\tiny{3}};

\draw (-1.75,0.25) -- (0.75,0.25);
\draw (-1.75,0.75) -- (0.75,0.75);
\draw (-1.75,-0.25) -- (0.75,-0.25);
\draw (-1.75,-0.75) -- (0.75,-0.75);
\draw (0.25,-0.75) -- (0.25,0.75);
\draw (0.75,-0.75) -- (0.75,0.75);
\draw (-0.25,-0.75) -- (-0.25,0.75);
\draw (-0.75,-0.75) -- (-0.75,0.75);
\draw (-1.25,-0.75) -- (-1.25,0.75);
\draw (-1.75,-0.75) -- (-1.75,0.75);
\draw[ultra thick, color=red] (-0.25,0.75) -- (-0.25,0.25);
\draw[ultra thick, color=red] (0.25,0.25) -- (-0.25,0.25);
\draw[ultra thick, color=red] (0.25,0.25) -- (0.25,-0.25);
\node at (0.25,0.25)[circle,draw=black,fill=black,scale=0.5]{};
\node at (-0.25,0.25)[circle,draw=black,fill=black,scale=0.5]{};
\node at (-0.25,0.75)[circle,draw=black,fill=black,scale=0.5]{};
\node at (0.25,-0.25)[circle,draw=black,fill=black,scale=0.5]{};

\node at (0.45,0.35) []{\tiny{2}};
\node at (-0.05,0.35) []{\tiny{3}};
\node at (-0.05,0.85) []{\tiny{2}};
\node at (0.45,-0.15) []{\tiny{0}};

\node at (0.35,0.05) []{\tiny{3}};
\node at (-0.15,0.05) []{\tiny{0}};
\node at (-0.15,0.55) []{\tiny{1}};
\node at (0.35,-0.45) []{\tiny{3}};

\node at (0.05,0.15) []{\tiny{0}};
\node at (-0.45,0.15) []{\tiny{1}};
\node at (-0.45,0.65) []{\tiny{0}};
\node at (0.05,-0.35) []{\tiny{2}};

\node at (0.15,0.45) []{\tiny{1}};
\node at (-0.35,0.45) []{\tiny{2}};
\node at (-0.35,0.95) []{\tiny{3}};
\node at (0.15,-0.05) []{\tiny{1}};

\draw (-1.75,0.25+5) -- (0.75,0.25+5);
\draw (-1.75,0.75+5) -- (0.75,0.75+5);
\draw (-1.75,-0.25+5) -- (0.75,-0.25+5);
\draw (-1.75,-0.75+5) -- (0.75,-0.75+5);
\draw (0.25,-0.75+5) -- (0.25,0.75+5);
\draw (0.75,-0.75+5) -- (0.75,0.75+5);
\draw (-0.25,-0.75+5) -- (-0.25,0.75+5);
\draw (-0.75,-0.75+5) -- (-0.75,0.75+5);
\draw (-1.25,-0.75+5) -- (-1.25,0.75+5);
\draw (-1.75,-0.75+5) -- (-1.75,0.75+5);
\draw[ultra thick, color=red] (-1.25,5.25) -- (0.25,5.25);
\draw[ultra thick, color=red] (-0.75,4.75) -- (0.25,4.75);
\draw[ultra thick, color=red] (-0.25,4.75) -- (-0.25,5.25);

\node at (-0.25,0.25+5)[circle,draw=black,fill=black,scale=0.5]{};
\node at (0.25,-0.25+5)[circle,draw=black,fill=black,scale=0.5]{};
\node at (-0.25,-0.25+5)[circle,draw=black,fill=black,scale=0.5]{};
\node at (0.25,0.25+5)[circle,draw=black,fill=black,scale=0.5]{};
\node at (-0.75,0.25+5)[circle,draw=black,fill=black,scale=0.5]{};
\node at (-1.25,0.25+5)[circle,draw=black,fill=black,scale=0.5]{};
\node at (-0.75,-0.25+5)[circle,draw=black,fill=black,scale=0.5]{};
\node at (0.45-1,0.35+5) []{\tiny{2}};
\node at (-0.05-1,0.35+5) []{\tiny{3}};
\node at (0.45-1,-0.15+5) []{\tiny{0}};
\node at (0.45,0.35+5) []{\tiny{1}};
\node at (-0.05,0.35+5) []{\tiny{3}};
\node at (0.45,-0.15+5) []{\tiny{1}};
\node at (0.45-0.5,0.35+4.5) []{\tiny{2}};

\node at (0.35-1,0.05+5) []{\tiny{3}};
\node at (-0.15-1,0.05+5) []{\tiny{2}};
\node at (0.35-1,-0.45+5) []{\tiny{3}};
\node at (0.35,0.05+5) []{\tiny{2}};
\node at (-0.15,0.05+5) []{\tiny{2}};
\node at (0.35,-0.45+5) []{\tiny{2}};
\node at (0.35-0.5,0.05+4.5) []{\tiny{1}};

\node at (0.05-1,0.15+5) []{\tiny{0}};
\node at (-0.45-1,0.15+5) []{\tiny{1}};
\node at (0.05-1,-0.35+5) []{\tiny{2}};
\node at (0.05,0.15+5) []{\tiny{3}};
\node at (-0.45,0.15+5) []{\tiny{1}};
\node at (0.05,-0.35+5) []{\tiny{3}};
\node at (0.05-0.5,0.15+4.5) []{\tiny{0}};

\node at (0.15-1,0.45+5) []{\tiny{1}};
\node at (-0.35-1,0.45+5) []{\tiny{0}};
\node at (0.15-1,-0.05+5) []{\tiny{1}};
\node at (0.15,0.45+5) []{\tiny{0}};
\node at (-0.35,0.45+5) []{\tiny{0}};
\node at (0.15,-0.05+5) []{\tiny{0}};
\node at (0.15-0.5,0.45+4.5) []{\tiny{3}};
\draw (-2,1+7.5) -- (1,1+7.5);
\draw (-2,0+7.5) -- (1,0+7.5);
\draw (-2,0.5+7.5) -- (1,0.5+7.5);
\draw (-2,-0.5+7.5) -- (1,-0.5+7.5);
\draw (-2,-1+7.5) -- (1,-1+7.5);
\draw (0,-1+7.5) -- (0,1+7.5);
\draw (0.5,-1+7.5) -- (0.5,1+7.5);
\draw (-0.5,-1+7.5) -- (-0.5,1+7.5);
\draw (1,-1+7.5) -- (1,1+7.5);
\draw (-1,-1+7.5) -- (-1,1+7.5);
\draw (-1.5,-1+7.5) -- (-1.5,1+7.5);
\draw (-2,-1+7.5) -- (-2,1+7.5);
\draw[ultra thick, color=red] (-1.5,2+5) -- (-0.5,2+5);
\draw[ultra thick, color=red] (-0.5,3+5) -- (-0.5,2+5);
\draw[ultra thick, color=red] (-1,3+5) -- (0,3+5);
\draw[ultra thick, color=red] (-1,3+5) -- (-1,2.5+5);

\node at (-0.5,2.5+5)[circle,draw=black,fill=black,scale=0.5]{};
\node at (-1.5,2+5)[circle,draw=black,fill=black,scale=0.5]{};
\node at (0,3+5)[circle,draw=black,fill=black,scale=0.5]{};
\node at (-0.5,3+5)[circle,draw=black,fill=black,scale=0.5]{};
\node at (-1,3+5)[circle,draw=black,fill=black,scale=0.5]{};
\node at (-1,2.5+5)[circle,draw=black,fill=black,scale=0.5]{};
\node at (-1,2+5)[circle,draw=black,fill=black,scale=0.5]{};
\node at (-0.5,2+5)[circle,draw=black,fill=black,scale=0.5]{};

\node at (0.2-0.5,3.1+4.5) []{\tiny{2}};
\node at (0.2,3.1+5) []{\tiny{2}};
\node at (0.2-1.5,3.1+4) []{\tiny{2}};
\node at (0.2-0.5,3.1+4) []{\tiny{3}};
\node at (0.2-0.5,3.1+5) []{\tiny{0}};
\node at (0.2-1,3.1+4) []{\tiny{0}};
\node at (0.2-1,3.1+4.5) []{\tiny{1}};
\node at (0.2-1,3.1+5) []{\tiny{1}};

\node at (0.1-0.5,2.8+4.5) []{\tiny{3}};
\node at (0.1,2.8+5) []{\tiny{1}};
\node at (0.1-1.5,2.8+4) []{\tiny{1}};
\node at (0.1-0.5,2.8+4) []{\tiny{0}};
\node at (0.1-0.5,2.8+5) []{\tiny{1}};
\node at (0.1-1,2.8+4) []{\tiny{3}};
\node at (0.1-1,2.8+4.5) []{\tiny{2}};
\node at (0.1-1,2.8+5) []{\tiny{0}};

\node at (-0.2-0.5,2.9+4.5) []{\tiny{0}};
\node at (-0.2,2.9+5) []{\tiny{0}};
\node at (-0.2-1.5,2.9+4) []{\tiny{0}};
\node at (-0.2-0.5,2.9+4) []{\tiny{1}};
\node at (-0.2-0.5,2.9+5) []{\tiny{2}};
\node at (-0.2-1,2.9+4) []{\tiny{2}};
\node at (-0.2-1,2.9+4.5) []{\tiny{3}};
\node at (-0.2-1,2.9+5) []{\tiny{3}};

\node at (-0.1-0.5,3.2+4.5) []{\tiny{1}};
\node at (-0.1,3.2+5) []{\tiny{3}};
\node at (-0.1-1.5,3.2+4) []{\tiny{3}};
\node at (-0.1-0.5,3.2+4) []{\tiny{2}};
\node at (-0.1-0.5,3.2+5) []{\tiny{3}};
\node at (-0.1-1,3.2+4) []{\tiny{1}};
\node at (-0.1-1,3.2+4.5) []{\tiny{0}};
\node at (-0.1-1,3.2+5) []{\tiny{2}};

\draw (-2+5,1+2.5) -- (1+5,1+2.5);
\draw (-2+5,0+2.5) -- (1+5,0+2.5);
\draw (-2+5,0.5+2.5) -- (1+5,0.5+2.5);
\draw (-2+5,-0.5+2.5) -- (1+5,-0.5+2.5);
\draw (-2+5,-1+2.5) -- (1+5,-1+2.5);
\draw (0+5,-1+2.5) -- (0+5,1+2.5);
\draw (0.5+5,-1+2.5) -- (0.5+5,1+2.5);
\draw (-0.5+5,-1+2.5) -- (-0.5+5,1+2.5);
\draw (1+5,-1+2.5) -- (1+5,1+2.5);
\draw (-1+5,-1+2.5) -- (-1+5,1+2.5);
\draw (-1.5+5,-1+2.5) -- (-1.5+5,1+2.5);
\draw (-2+5,-1+2.5) -- (-2+5,1+2.5);

\draw[ultra thick, color=red] (-1+5,2.5) -- (0.5+5,2.5);
\draw[ultra thick, color=red] (-1+5,2) -- (-1+5,3);
\draw[ultra thick, color=red] (-0.5+5,2) -- (-0.5+5,3);
\draw[ultra thick, color=red] (0+5,2) -- (0+5,3);

\node at (0+5,2.5)[regular polygon, regular polygon sides=4,draw=black,fill=red,scale=0.5]{};
\node at (-0.5+5,2.5)[circle,draw=black,fill=black,scale=0.5]{};
\node at (0.5+5,2.5)[circle,draw=black,fill=black,scale=0.5]{};
\node at (0+5,2)[circle,draw=black,fill=black,scale=0.5]{};
\node at (0+5,3)[circle,draw=black,fill=black,scale=0.5]{};
\node at (-0.5+5,3)[circle,draw=black,fill=black,scale=0.5]{};
\node at (-1+5,3)[circle,draw=black,fill=black,scale=0.5]{};
\node at (-1+5,2.5)[circle,draw=black,fill=black,scale=0.5]{};
\node at (-1+5,2)[circle,draw=black,fill=black,scale=0.5]{};
\node at (-0.5+5,2)[circle,draw=black,fill=black,scale=0.5]{};

\node at (0.2-0.5+5,3.1-0.5) []{\tiny{0}};
\node at (0.7-0.5+5,2.6-0.5) []{\tiny{0}};
\node at (0.7-0.5+5,3.6-0.5) []{\tiny{0}};
\node at (0.7-0.5+5,3.1-0.5) []{\tiny{0}};
\node at (0.2+0.5+5,3.1-0.5) []{\tiny{0}};
\node at (0.7-1.5+5,2.6-0.5) []{\tiny{0}};
\node at (0.7-1.5+5,3.6-0.5) []{\tiny{0}};
\node at (0.7-1.5+5,3.1-0.5) []{\tiny{0}};
\node at (0.2-0.5+5,3.1) []{\tiny{0}};
\node at (0.2-0.5+5,3.1-1) []{\tiny{0}};

\node at (0.1-0.5+5,2.8-0.5) []{\tiny{1}};
\node at (0.6-0.5+5,2.3-0.5) []{\tiny{1}};
\node at (0.6-0.5+5,3.3-0.5) []{\tiny{1}};
\node at (0.6-0.5+5,2.8-0.5) []{\tiny{1}};
\node at (0.1+0.5+5,2.8-0.5) []{\tiny{1}};
\node at (0.6-1.5+5,2.3-0.5) []{\tiny{1}};
\node at (0.6-1.5+5,3.3-0.5) []{\tiny{1}};
\node at (0.6-1.5+5,2.8-0.5) []{\tiny{1}};
\node at (0.1-0.5+5,2.8-1) []{\tiny{1}};
\node at (0.1-0.5+5,2.8) []{\tiny{1}};

\node at (-0.2-0.5+5,2.9-0.5) []{\tiny{2}};
\node at (0.3-0.5+5,2.4-0.5) []{\tiny{2}};
\node at (0.3-0.5+5,3.4-0.5) []{\tiny{2}};
\node at (0.3-0.5+5,2.9-0.5) []{\tiny{2}};
\node at (-0.2+0.5+5,2.9-0.5) []{\tiny{2}};
\node at (0.3-1.5+5,2.4-0.5) []{\tiny{2}};
\node at (0.3-1.5+5,3.4-0.5) []{\tiny{2}};
\node at (0.3-1.5+5,2.9-0.5) []{\tiny{2}};
\node at (-0.2-0.5+5,2.9-0) []{\tiny{2}};
\node at (-0.2-0.5+5,2.9-1) []{\tiny{2}};

\node at (-0.1-0.5+5,3.2-0.5) []{\tiny{3}};
\node at (0.4-0.5+5,2.7-0.5) []{\tiny{3}};
\node at (0.4-0.5+5,3.7-0.5) []{\tiny{3}};
\node at (0.4-0.5+5,3.2-0.5) []{\tiny{3}};
\node at (-0.1+0.5+5,3.2-0.5) []{\tiny{3}};
\node at (0.4-1.5+5,2.7-0.5) []{\tiny{3}};
\node at (0.4-1.5+5,3.7-0.5) []{\tiny{3}};
\node at (0.4-1.5+5,3.2-0.5) []{\tiny{3}};
\node at (-0.1-0.5+5,3.2) []{\tiny{3}};
\node at (-0.1-0.5+5,3.2-1) []{\tiny{3}};

\draw (-1.75+5,0.25) -- (0.75+5,0.25);
\draw (-1.75+5,0.75) -- (0.75+5,0.75);
\draw (-1.75+5,-0.25) -- (0.75+5,-0.25);
\draw (-1.75+5,-0.75) -- (0.75+5,-0.75);
\draw (0.25+5,-0.75) -- (0.25+5,0.75);
\draw (0.75+5,-0.75) -- (0.75+5,0.75);
\draw (-0.25+5,-0.75) -- (-0.25+5,0.75);
\draw (-0.75+5,-0.75) -- (-0.75+5,0.75);
\draw (-1.25+5,-0.75) -- (-1.25+5,0.75);
\draw (-1.75+5,-0.75) -- (-1.75+5,0.75);
\draw[ultra thick, color=red] (-0.25+5,0.75) -- (-0.25+5,0.25);
\draw[ultra thick, color=red] (0.25+5,0.25) -- (-0.25+5,0.25);
\draw[ultra thick, color=red] (0.25+5,0.25) -- (0.25+5,-0.25);
\node at (0.25+5,0.25)[circle,draw=black,fill=black,scale=0.5]{};
\node at (-0.25+5,0.25)[circle,draw=black,fill=black,scale=0.5]{};
\node at (-0.25+5,0.75)[circle,draw=black,fill=black,scale=0.5]{};
\node at (0.25+5,-0.25)[circle,draw=black,fill=black,scale=0.5]{};

\node at (0.45+5,0.35) []{\tiny{0}};
\node at (-0.05+5,0.35) []{\tiny{0}};
\node at (-0.05+5,0.85) []{\tiny{0}};
\node at (0.45+5,-0.15) []{\tiny{0}};

\node at (0.35+5,0.05) []{\tiny{1}};
\node at (-0.15+5,0.05) []{\tiny{1}};
\node at (-0.15+5,0.55) []{\tiny{1}};
\node at (0.35+5,-0.45) []{\tiny{1}};

\node at (0.05+5,0.15) []{\tiny{2}};
\node at (-0.45+5,0.15) []{\tiny{2}};
\node at (-0.45+5,0.65) []{\tiny{2}};
\node at (0.05+5,-0.35) []{\tiny{2}};

\node at (0.15+5,0.45) []{\tiny{3}};
\node at (-0.35+5,0.45) []{\tiny{3}};
\node at (-0.35+5,0.95) []{\tiny{3}};
\node at (0.15+5,-0.05) []{\tiny{3}};

\draw (-1.75+5,0.25+5) -- (0.75+5,0.25+5);
\draw (-1.75+5,0.75+5) -- (0.75+5,0.75+5);
\draw (-1.75+5,-0.25+5) -- (0.75+5,-0.25+5);
\draw (-1.75+5,-0.75+5) -- (0.75+5,-0.75+5);
\draw (0.25+5,-0.75+5) -- (0.25+5,0.75+5);
\draw (0.75+5,-0.75+5) -- (0.75+5,0.75+5);
\draw (-0.25+5,-0.75+5) -- (-0.25+5,0.75+5);
\draw (-0.75+5,-0.75+5) -- (-0.75+5,0.75+5);
\draw (-1.25+5,-0.75+5) -- (-1.25+5,0.75+5);
\draw (-1.75+5,-0.75+5) -- (-1.75+5,0.75+5);
\draw[ultra thick, color=red] (-1.25+5,5.25) -- (0.25+5,5.25);
\draw[ultra thick, color=red] (-0.75+5,4.75) -- (0.25+5,4.75);
\draw[ultra thick, color=red] (-0.25+5,4.75) -- (-0.25+5,5.25);

\node at (-0.25+5,0.25+5)[circle,draw=black,fill=black,scale=0.5]{};
\node at (0.25+5,-0.25+5)[circle,draw=black,fill=black,scale=0.5]{};
\node at (-0.25+5,-0.25+5)[circle,draw=black,fill=black,scale=0.5]{};
\node at (0.25+5,0.25+5)[circle,draw=black,fill=black,scale=0.5]{};
\node at (-0.75+5,0.25+5)[circle,draw=black,fill=black,scale=0.5]{};
\node at (-1.25+5,0.25+5)[circle,draw=black,fill=black,scale=0.5]{};
\node at (-0.75+5,-0.25+5)[circle,draw=black,fill=black,scale=0.5]{};
\node at (0.45-1+5,0.35+5) []{\tiny{0}};
\node at (-0.05-1+5,0.35+5) []{\tiny{0}};
\node at (0.45-1+5,-0.15+5) []{\tiny{0}};
\node at (0.45+5,0.35+5) []{\tiny{0}};
\node at (-0.05+5,0.35+5) []{\tiny{0}};
\node at (0.45+5,-0.15+5) []{\tiny{0}};
\node at (0.45-0.5+5,0.35+4.5) []{\tiny{0}};

\node at (0.35-1+5,0.05+5) []{\tiny{1}};
\node at (-0.15-1+5,0.05+5) []{\tiny{1}};
\node at (0.35-1+5,-0.45+5) []{\tiny{1}};
\node at (0.35+5,0.05+5) []{\tiny{1}};
\node at (-0.15+5,0.05+5) []{\tiny{1}};
\node at (0.35+5,-0.45+5) []{\tiny{1}};
\node at (0.35-0.5+5,0.05+4.5) []{\tiny{1}};

\node at (0.05-1+5,0.15+5) []{\tiny{2}};
\node at (-0.45-1+5,0.15+5) []{\tiny{2}};
\node at (0.05-1+5,-0.35+5) []{\tiny{2}};
\node at (0.05+5,0.15+5) []{\tiny{2}};
\node at (-0.45+5,0.15+5) []{\tiny{2}};
\node at (0.05+5,-0.35+5) []{\tiny{2}};
\node at (0.05-0.5+5,0.15+4.5) []{\tiny{2}};

\node at (0.15-1+5,0.45+5) []{\tiny{3}};
\node at (-0.35-1+5,0.45+5) []{\tiny{3}};
\node at (0.15-1+5,-0.05+5) []{\tiny{3}};
\node at (0.15+5,0.45+5) []{\tiny{3}};
\node at (-0.35+5,0.45+5) []{\tiny{3}};
\node at (0.15+5,-0.05+5) []{\tiny{3}};
\node at (0.15-0.5+5,0.45+4.5) []{\tiny{3}};
\draw (-2+5,1+7.5) -- (1+5,1+7.5);
\draw (-2+5,0+7.5) -- (1+5,0+7.5);
\draw (-2+5,0.5+7.5) -- (1+5,0.5+7.5);
\draw (-2+5,-0.5+7.5) -- (1+5,-0.5+7.5);
\draw (-2+5,-1+7.5) -- (1+5,-1+7.5);
\draw (0+5,-1+7.5) -- (0+5,1+7.5);
\draw (0.5+5,-1+7.5) -- (0.5+5,1+7.5);
\draw (-0.5+5,-1+7.5) -- (-0.5+5,1+7.5);
\draw (1+5,-1+7.5) -- (1+5,1+7.5);
\draw (-1+5,-1+7.5) -- (-1+5,1+7.5);
\draw (-1.5+5,-1+7.5) -- (-1.5+5,1+7.5);
\draw (-2+5,-1+7.5) -- (-2+5,1+7.5);
\draw[ultra thick, color=red] (-1.5+5,2+5) -- (-0.5+5,2+5);
\draw[ultra thick, color=red] (-0.5+5,3+5) -- (-0.5+5,2+5);
\draw[ultra thick, color=red] (-1+5,3+5) -- (0+5,3+5);
\draw[ultra thick, color=red] (-1+5,3+5) -- (-1+5,2.5+5);

\node at (-0.5+5,2.5+5)[circle,draw=black,fill=black,scale=0.5]{};
\node at (-1.5+5,2+5)[circle,draw=black,fill=black,scale=0.5]{};
\node at (0+5,3+5)[circle,draw=black,fill=black,scale=0.5]{};
\node at (-0.5+5,3+5)[circle,draw=black,fill=black,scale=0.5]{};
\node at (-1+5,3+5)[circle,draw=black,fill=black,scale=0.5]{};
\node at (-1+5,2.5+5)[circle,draw=black,fill=black,scale=0.5]{};
\node at (-1+5,2+5)[circle,draw=black,fill=black,scale=0.5]{};
\node at (-0.5+5,2+5)[circle,draw=black,fill=black,scale=0.5]{};

\node at (0.2-0.5+5,3.1+4.5) []{\tiny{0}};
\node at (0.2+5,3.1+5) []{\tiny{0}};
\node at (0.2-1.5+5,3.1+4) []{\tiny{0}};
\node at (0.2-0.5+5,3.1+4) []{\tiny{0}};
\node at (0.2-0.5+5,3.1+5) []{\tiny{0}};
\node at (0.2-1+5,3.1+4) []{\tiny{0}};
\node at (0.2-1+5,3.1+4.5) []{\tiny{0}};
\node at (0.2-1+5,3.1+5) []{\tiny{0}};

\node at (0.1-0.5+5,2.8+4.5) []{\tiny{1}};
\node at (0.1+5,2.8+5) []{\tiny{1}};
\node at (0.1-1.5+5,2.8+4) []{\tiny{1}};
\node at (0.1-0.5+5,2.8+4) []{\tiny{1}};
\node at (0.1-0.5+5,2.8+5) []{\tiny{1}};
\node at (0.1-1+5,2.8+4) []{\tiny{1}};
\node at (0.1-1+5,2.8+4.5) []{\tiny{1}};
\node at (0.1-1+5,2.8+5) []{\tiny{1}};

\node at (-0.2-0.5+5,2.9+4.5) []{\tiny{2}};
\node at (-0.2+5,2.9+5) []{\tiny{2}};
\node at (-0.2-1.5+5,2.9+4) []{\tiny{2}};
\node at (-0.2-0.5+5,2.9+4) []{\tiny{2}};
\node at (-0.2-0.5+5,2.9+5) []{\tiny{2}};
\node at (-0.2-1+5,2.9+4) []{\tiny{2}};
\node at (-0.2-1+5,2.9+4.5) []{\tiny{2}};
\node at (-0.2-1+5,2.9+5) []{\tiny{2}};

\node at (-0.1-0.5+5,3.2+4.5) []{\tiny{3}};
\node at (-0.1+5,3.2+5) []{\tiny{3}};
\node at (-0.1-1.5+5,3.2+4) []{\tiny{3}};
\node at (-0.1-0.5+5,3.2+4) []{\tiny{3}};
\node at (-0.1-0.5+5,3.2+5) []{\tiny{3}};
\node at (-0.1-1+5,3.2+4) []{\tiny{3}};
\node at (-0.1-1+5,3.2+4.5) []{\tiny{3}};
\node at (-0.1-1+5,3.2+5) []{\tiny{3}};

\draw (-2+10,1+2.5) -- (1+10,1+2.5);
\draw (-2+10,0+2.5) -- (1+10,0+2.5);
\draw (-2+10,0.5+2.5) -- (1+10,0.5+2.5);
\draw (-2+10,-0.5+2.5) -- (1+10,-0.5+2.5);
\draw (-2+10,-1+2.5) -- (1+10,-1+2.5);
\draw (0+10,-1+2.5) -- (0+10,1+2.5);
\draw (0.5+10,-1+2.5) -- (0.5+10,1+2.5);
\draw (-0.5+10,-1+2.5) -- (-0.5+10,1+2.5);
\draw (1+10,-1+2.5) -- (1+10,1+2.5);
\draw (-1+10,-1+2.5) -- (-1+10,1+2.5);
\draw (-1.5+10,-1+2.5) -- (-1.5+10,1+2.5);
\draw (-2+10,-1+2.5) -- (-2+10,1+2.5);

\draw[ultra thick, color=red] (-1+10,2.5) -- (0.5+10,2.5);
\draw[ultra thick, color=red] (-1+10,2) -- (-1+10,3);
\draw[ultra thick, color=red] (-0.5+10,2) -- (-0.5+10,3);
\draw[ultra thick, color=red] (0+10,2) -- (0+10,3);

\node at (0+10,2.5)[regular polygon, regular polygon sides=4,draw=black,fill=red,scale=0.5]{};
\node at (-0.5+10,2.5)[circle,draw=black,fill=black,scale=0.5]{};
\node at (0.5+10,2.5)[circle,draw=black,fill=black,scale=0.5]{};
\node at (0+10,2)[circle,draw=black,fill=black,scale=0.5]{};
\node at (0+10,3)[circle,draw=black,fill=black,scale=0.5]{};
\node at (-0.5+10,3)[circle,draw=black,fill=black,scale=0.5]{};
\node at (-1+10,3)[circle,draw=black,fill=black,scale=0.5]{};
\node at (-1+10,2.5)[circle,draw=black,fill=black,scale=0.5]{};
\node at (-1+10,2)[circle,draw=black,fill=black,scale=0.5]{};
\node at (-0.5+10,2)[circle,draw=black,fill=black,scale=0.5]{};
\node at (10.7,2.7) []{1};
\node at (10.2,2.7) []{0};
\node at (9.7,2.7) []{4};
\node at (9.2,2.7) []{3};
\node at (10.2,3.2) []{2};
\node at (9.7,3.2) []{1};
\node at (9.2,3.2) []{0};
\node at (10.2,2.2) []{3};
\node at (9.7,2.2) []{2};
\node at (9.2,2.2) []{1};

\draw (-1.75+10,0.25) -- (0.75+10,0.25);
\draw (-1.75+10,0.75) -- (0.75+10,0.75);
\draw (-1.75+10,-0.25) -- (0.75+10,-0.25);
\draw (-1.75+10,-0.75) -- (0.75+10,-0.75);
\draw (0.25+10,-0.75) -- (0.25+10,0.75);
\draw (0.75+10,-0.75) -- (0.75+10,0.75);
\draw (-0.25+10,-0.75) -- (-0.25+10,0.75);
\draw (-0.75+10,-0.75) -- (-0.75+10,0.75);
\draw (-1.25+10,-0.75) -- (-1.25+10,0.75);
\draw (-1.75+10,-0.75) -- (-1.75+10,0.75);
\draw[ultra thick, color=red] (-0.25+10,0.75) -- (-0.25+10,0.25);
\draw[ultra thick, color=red] (0.25+10,0.25) -- (-0.25+10,0.25);
\draw[ultra thick, color=red] (0.25+10,0.25) -- (0.25+10,-0.25);
\node at (0.25+10,0.25)[circle,draw=black,fill=black,scale=0.5]{};
\node at (-0.25+10,0.25)[circle,draw=black,fill=black,scale=0.5]{};
\node at (-0.25+10,0.75)[circle,draw=black,fill=black,scale=0.5]{};
\node at (0.25+10,-0.25)[circle,draw=black,fill=black,scale=0.5]{};

\node at (10.45,0.25+0.2) []{11};
\node at (10.45,-0.25+0.2) []{14};
\node at (9.95,0.25+0.2) []{10};
\node at (9.95,0.75+0.2) []{12};

\draw (-1.75+10,0.25+5) -- (0.75+10,0.25+5);
\draw (-1.75+10,0.75+5) -- (0.75+10,0.75+5);
\draw (-1.75+10,-0.25+5) -- (0.75+10,-0.25+5);
\draw (-1.75+10,-0.75+5) -- (0.75+10,-0.75+5);
\draw (0.25+10,-0.75+5) -- (0.25+10,0.75+5);
\draw (0.75+10,-0.75+5) -- (0.75+10,0.75+5);
\draw (-0.25+10,-0.75+5) -- (-0.25+10,0.75+5);
\draw (-0.75+10,-0.75+5) -- (-0.75+10,0.75+5);
\draw (-1.25+10,-0.75+5) -- (-1.25+10,0.75+5);
\draw (-1.75+10,-0.75+5) -- (-1.75+10,0.75+5);
\draw[ultra thick, color=red] (-1.25+10,5.25) -- (0.25+10,5.25);
\draw[ultra thick, color=red] (-0.75+10,4.75) -- (0.25+10,4.75);
\draw[ultra thick, color=red] (-0.25+10,4.75) -- (-0.25+10,5.25);

\node at (-0.25+10,0.25+5)[circle,draw=black,fill=black,scale=0.5]{};
\node at (0.25+10,-0.25+5)[circle,draw=black,fill=black,scale=0.5]{};
\node at (-0.25+10,-0.25+5)[circle,draw=black,fill=black,scale=0.5]{};
\node at (0.25+10,0.25+5)[circle,draw=black,fill=black,scale=0.5]{};
\node at (-0.75+10,0.25+5)[circle,draw=black,fill=black,scale=0.5]{};
\node at (-1.25+10,0.25+5)[circle,draw=black,fill=black,scale=0.5]{};
\node at (-0.75+10,-0.25+5)[circle,draw=black,fill=black,scale=0.5]{};

\node at (10.45,0.25+5.2) []{6};
\node at (9.95,0.25+5.2) []{5};
\node at (9.45,0.25+5.2) []{9};
\node at (8.95,0.25+5.2) []{8};
\node at (10.45,-0.25+5.2) []{9};
\node at (9.95,-0.25+5.2) []{8};
\node at (9.45,-0.25+5.2) []{7};

\draw (-2+10,1+7.5) -- (1+10,1+7.5);
\draw (-2+10,0+7.5) -- (1+10,0+7.5);
\draw (-2+10,0.5+7.5) -- (1+10,0.5+7.5);
\draw (-2+10,-0.5+7.5) -- (1+10,-0.5+7.5);
\draw (-2+10,-1+7.5) -- (1+10,-1+7.5);
\draw (0+10,-1+7.5) -- (0+10,1+7.5);
\draw (0.5+10,-1+7.5) -- (0.5+10,1+7.5);
\draw (-0.5+10,-1+7.5) -- (-0.5+10,1+7.5);
\draw (1+10,-1+7.5) -- (1+10,1+7.5);
\draw (-1+10,-1+7.5) -- (-1+10,1+7.5);
\draw (-1.5+10,-1+7.5) -- (-1.5+10,1+7.5);
\draw (-2+10,-1+7.5) -- (-2+10,1+7.5);
\draw[ultra thick, color=red] (-1.5+10,2+5) -- (-0.5+10,2+5);
\draw[ultra thick, color=red] (-0.5+10,3+5) -- (-0.5+10,2+5);
\draw[ultra thick, color=red] (-1+10,3+5) -- (0+10,3+5);
\draw[ultra thick, color=red] (-1+10,3+5) -- (-1+10,2.5+5);

\node at (-0.5+10,2.5+5)[circle,draw=black,fill=black,scale=0.5]{};
\node at (-1.5+10,2+5)[circle,draw=black,fill=black,scale=0.5]{};
\node at (0+10,3+5)[circle,draw=black,fill=black,scale=0.5]{};
\node at (-0.5+10,3+5)[circle,draw=black,fill=black,scale=0.5]{};
\node at (-1+10,3+5)[circle,draw=black,fill=black,scale=0.5]{};
\node at (-1+10,2.5+5)[circle,draw=black,fill=black,scale=0.5]{};
\node at (-1+10,2+5)[circle,draw=black,fill=black,scale=0.5]{};
\node at (-0.5+10,2+5)[circle,draw=black,fill=black,scale=0.5]{};

\node at (9.2,2.7+5) []{13};
\node at (10.2,3.2+5) []{12};
\node at (9.7,3.2+5) []{11};
\node at (9.2,3.2+5) []{10};
\node at (9.7,2.7+5) []{14};
\node at (9.7,2.2+5) []{12};
\node at (9.2,2.2+5) []{11};
\node at (8.7,2.2+5) []{10};
\end{tikzpicture}
\end{center}
\caption{Edges of a spanning tree of particles, a possible numbering of the ports of the particles before (Figure~\ref{portreconfig}.a) and after the execution of Algorithm~\ref{alg5} (Figure~\ref{portreconfig}.b) and the 2-identifier obtained by executing Algorithm~\ref{alg4} (Figure~\ref{portreconfig}.c) in $\mathcal{F}$ (square: leader; thick line: edge of the spanning tree; small number: port number of a particle; big number: 4-identifier of a particle).}
\label{portreconfig}

\end{figure}

\begin{prop}
At the end of the execution of Algorithm \ref{alg4}, any two particles at distance at most $\ell$ have two different $\ell$-local-identifiers.
\end{prop}
\begin{proof}
We recall that the distance between two particles $p$ and $p'$ at position $(i,j,k)$ and $(i',j',k')$ of $\mathcal{F}$ is given by the following formula \cite{GT2019}:
$$d_{\mathcal{F}}((i,j,k),(i',j',k'))=p_{+}\left(|i-i'|-\frac{|k-k'|}{2}\right)+p_{+}\left(|j-j'|-\frac{|k-k'|}{2}\right)+ |k-k'|.$$
Suppose that $|k-k'|>\ell $, by the given formula, we obtain that $d(p,p')> \ell$.
Thus, it implies that any two particles $p$ and $p'$ at distance at most $\ell$ are such that $|k-k'|\le \ell $. 

First, suppose that $k\neq k'$.
For any two layers $k$ and $k'$, $f_\ell (i,j,k)$ has its value between $mod(k, \ell+1) m_{\ell}$ and $mod(k, \ell+1) m_{\ell}+m_{\ell}-1$ and $f_\ell (i',j',k')$ has its value between $mod(k', \ell+1) m_{\ell}$ and $mod(k', \ell+1) m_{\ell}+m_{\ell}-1$. It implies that $f_\ell (i,j,k)\neq f_\ell (i',j',k')$.

Second, suppose that $k=k'$ and suppose, without loss of generality, that $k$ is even. Also, in order to simplify the proof, we can consider that $i=0$, $j=0$, $k=0$ (we can suppose this, since we can translate the structure in order that $p$ is positioned where the origin $(0,0,0)$ is). 
Since $p$ and $p'$ are at distance at most $\ell$, we have $|i'|+|j'|\le \ell$. 
Consequently, $0< f_\ell (i',j',k')\le \ell^2<m_{\ell}$ or $-m_{\ell}<-\ell^2 \le f_\ell (i',j',k') <0$ ($\ell^2<m_{\ell}$ being true when $\ell$ is a positive integer).
Since  $f_\ell (0,0,0)=0$, we are sure that  $f_\ell (i',j',k')\neq 0$.
\end{proof}

Moreover, about the behavior of Algorithms \ref{alg5} or \ref{alg4}, we state the following.
\begin{prop}
Whatever the structure of $\mathcal{F}[S]$, if $\mathcal{F}[S]$ is connected, then after $ diam(\mathcal{F}[S])$ rounds, both Algorithms \ref{alg5} or \ref{alg4} have finished their tasks. Moreover, the memory space used by Algorithm \ref{alg5} is constant and the memory space used by Algorithm \ref{alg4} is at most $O(\log(\ell^3))$.
\end{prop}
\begin{proof}
First, we state that both Algorithm \ref{alg5} and Algorithm \ref{alg4} finish after $diam(\mathcal{F}[S])$ rounds. That is the case since the height of a spanning tree of $\mathcal{F}[S]$ is bounded by $diam(\mathcal{F}[S])$ and the two algorithms consist in transmitting message through the spanning tree. Note that this implies that the number of sent messages in each algorithm is bounded by $|S|$.
Second, since the maximum degree is bounded in $\mathcal{F}$, the required memory space in order to store the ports of the children and the parent is constant. However, a $O(\log(\ell^3))$ memory space is required to store the $\ell$-local-identifier.
\end{proof}

\subsection{Global identifiers}
Computing global identifier in a distributed way is easy as soon as a leader has been elected. One can use the relative position of the particle to the leader as an identifier. To this end, we re-use the functions $I$, $J$ and $K$ defined in Section~\ref{homo}.
As for local identifiers, our algorithm in order to compute the global identifiers works in four steps as follows:
First, a leader election algorithm is used to have a unique particle in state \textbf{L} and the others in state \textbf{N} (Algorithm 1 or 2).
Second, computing a spanning tree with a distributed algorithm.
Third, changing the way the port are numbered in order that every particle has its ports numbered by the same number going in the same cardinal direction in $\mathcal{F}$, using Algorithm~\ref{alg5} (it is needed in the heterogeneous case).
Fourth, give the identifier $(0,0,0)$ to the particle in state \textbf{L} and, for a particle $p$ having identifier $(i,j,k)$, the inductive step consists in using messages to send the identifier $(I(i,a),J(j,a),K(k,a))$ to the neighbor of $p$ connected through port $a$ of $p$ (see Algorithm \ref{alg3}).

Similarly than in algorithms~\ref{alg5} and~\ref{alg4}, in Algorithm~\ref{alg3}, we suppose that, for each particle $p$, the set of ports $\text{child}(p)$ contains the port numbers of the particles in communication with its children in the spanning tree.

\begin{algorithm}
\caption{The global identifier $(i,j,k)$ algorithm for a particle $p$.} 
\label{alg3}
\begin{algorithmic} 
\State \textbf{Case 1: } State \textbf{L} \\ Set $i=0$, $j=0$ and $k=0$\\ For every port $a$ from $\text{child}(p)$, send the message $(I(i,a),J(j,a),K(k,a))$ to the neighbor of $p$ connected through port $a$

\State \textbf{Case 2: } State \textbf{N}
\If {message $(i',j',k')$ is received through port $a$}
        \State Set $i=i'$, $j=j'$ and $k=k'$
	\State For every port $a$ from $\text{child}(p)$, send the message $(I(i,a),J(j,a),K(k,a))$ to the neighbor of $p$ connected through port $a$
\EndIf
\end{algorithmic}
\end{algorithm}

\begin{prop}
Whatever the structure of $\mathcal{F}[S]$,  if $\mathcal{F}[S]$ is connected, then after $ diam(\mathcal{F}[S])$ rounds, Algorithm \ref{alg3} has finished to compute the global identifier. Moreover, the memory space used by Algorithm \ref{alg3} is at most $O(\log(|S|))$.
\end{prop}
\begin{proof}
First, Algorithm \ref{alg3} finish after $diam(\mathcal{F}[S])$ rounds since the height of a spanning tree of $\mathcal{F}[S]$ is bounded by $diam(\mathcal{F}[S])$ and the algorithm consists in transmitting message through the spanning tree. Note that this implies that the number of sent messages is bounded by $|S|$.
Second, since the global identifier $(i,j,k)$ is such that $|i|\le |S|$, $|j|\le |S|$ and $|k|\le |S|$,  a $O(\log(|S|))$ memory space is sufficient to store the global identifier.
\end{proof}
Moreover, about the behavior of Algorithm \ref{alg3}, we state the following.
\begin{prop}
After the execution of Algorithm \ref{alg3}, each particle has an unique global identifier.
\end{prop}
\begin{proof}
Suppose that two particle $p$ and $p'$ have the same global identifier. Since the global identifier corresponds to the position of a particle relatively to the leader particle, and since there is an unique leader particle, it implies that $p$ and $p'$ have the same position. 
\end{proof}

\section{Concluding remarks}
In this paper, we have presented two new leader election algorithms based on local computation for programmable matter in the 3-dimensional space. The first one considers that all particles have their ports distributed in the same directions (homogeneous case) but works with arbitrary shape while the second one does not require any initial port direction configuration (heterogeneous case) but works only for some shapes including the three dimensional sphere and cube. We have also presented an algorithm, which affects identifiers to the particles such that every two particles at distance at most $\ell$ have different identifiers.  Finally, we have presented an algorithm affecting a unique global identifier to each particle. Algorithms \ref{alg2}, \ref{alg5} and \ref{alg4} only require a $O(1)$-space memory (for fixed $\ell$), hence are well suited for programmable matter in which particles have constant memory capacities.

As future work, it would be interesting to extend our leader election algorithm in order it works also for sets of particles forming more general shapes than the ones considered in the heterogeneous case. Another interesting question could be to use our results for clustering the set of particles in several sets, which induce subgraphs of small diameter.

\section*{Acknowledgments}
This work is done in the context of the B3PM project and was supported by the French "Investissements d'Avenir" program, project ISITE-BFC (contract ANR-15-IDEX-03).

\end{document}